\documentclass[sigconf,screen,balance=false]{acmart}

\newif\ifanonymous
\anonymousfalse
\relax

\AtBeginDocument{}

\usepackage[utf8]{inputenc} 
\usepackage[T1]{fontenc}    
\usepackage{xcolor}         
\definecolor{linkdarkblue}{rgb}{0, 0.08, 0.45}    
\usepackage{url}            
\usepackage{array}          
\usepackage{colortbl}       
\usepackage{makecell}       
\usepackage{tabularx}       
\usepackage{bigstrut}       
\usepackage{multirow}       
\usepackage{colortbl}       
\usepackage{nicefrac}       
\usepackage{microtype}      
\usepackage{graphicx}       
\usepackage{bm}             
\usepackage[capitalise,noabbrev]{cleveref} 
\usepackage{subcaption}     
\usepackage{booktabs}       
\usepackage{wrapfig}        
\usepackage{pifont}         
\usepackage{marvosym}       
\usepackage{algorithm}      
\usepackage{mathtools}      
\usepackage{enumitem}       
\usepackage{algorithm}      
\usepackage{algorithmicx}   
\usepackage{algpseudocode}  
\usepackage{thm-restate}    
\usepackage{fontawesome}    
\usepackage{comment}        
\usepackage[normalem]{ulem} 
\usepackage[most]{tcolorbox}
\usepackage{xspace}         
\usepackage{amsmath,amsfonts,bm}

\def\eqref#1{equation~\ref{#1}}

\def\1{\bm{1}}

\DeclareMathAlphabet{\mathsfit}{\encodingdefault}{\sfdefault}{m}{sl}
\SetMathAlphabet{\mathsfit}{bold}{\encodingdefault}{\sfdefault}{bx}{n}

\allowdisplaybreaks

\newcommand{\noindentparnolinenospace}[1]{\noindent\textbf{#1} }

\newcommand{\noindentparnoline}[1]{\addvspace{0.50\baselineskip}\noindent\textbf{#1} }

\newcommand{\eqnref}[1]{(\ref{#1})}

\numberwithin{equation}{section}    
\numberwithin{algorithm}{section}   

\theoremstyle{plain}

\newtheorem*{theorem*}{Theorem}

\newtheorem*{lemma*}{Lemma}

\newtheorem*{axiom*}{Axiom}

\newtheorem*{conjecture*}{Conjecture}

\newtheorem*{assumption*}{Assumption}

\newtheorem*{proposition*}{Proposition}

\newtheorem*{corollary*}{Corollary}

\theoremstyle{definition}

\newtheorem*{definition*}{Definition}

\newtheorem*{example*}{Example}

\newtheorem*{problem*}{Problem}

\newtheorem*{question*}{Question}

\newtheorem*{exercise*}{Exercise}

\newtheorem*{remark*}{Remark}

\definecolor{darkgreen}{rgb}{0.0, 0.5, 0.0}
\definecolor{skyblue}{rgb}{0.53, 0.81, 0.98}
\definecolor{beamsearchblue}{rgb}{0.42, 0.56, 0.75}     
\definecolor{beamsearchyellow}{rgb}{0.84, 0.71, 0.34}   
\definecolor{beamsearchred}{rgb}{0.72, 0.33, 0.31}      
\definecolor{darkred}{rgb}{0.6, 0, 0}
\definecolor{pruning_red}{rgb}{0.64, 0.12, 0.21}        
\definecolor{pruning_gray}{rgb}{0.6, 0.6, 0.6}          
\definecolor{sensitivity_blue}{rgb}{0.30, 0.45, 0.69}   
\definecolor{sensitivity_orange}{rgb}{0.87, 0.52, 0.32} 

\newcolumntype{Y}{>{\centering\arraybackslash}X}

\newcommand{\ie}{\text{i.e., }}         
\newcommand{\eg}{\text{e.g., }}         
\newcommand{\wrt}{\text{w.r.t. }}       
\newcommand{\cf}{\text{cf. }}           

\newtcolorbox{statementbox}[1][]
{
    enhanced,
    colback=blue!5!white,
    colframe=blue!60!black,
    coltitle=white,
    fonttitle=\bfseries,
    boxrule=0pt,
    bottomrule=1.5pt,
    arc=0pt,
    auto outer arc,
    left=8pt, right=8pt, top=5pt, bottom=5pt,
    toptitle=2pt, bottomtitle=2pt,
    before skip=8pt, after skip=8pt,
    breakable,
    #1
}

\definecolor{bocchipink}{HTML}{F3A6B2}
\definecolor{nijikayellow}{HTML}{F8DD74}
\definecolor{ryoblue}{HTML}{6680AE}
\definecolor{kitared}{HTML}{DD6D63}

\tcbset{
    kessokustyle/.style={
        enhanced, breakable,
        colback=#1!8!white,
        colframe=#1!90!black,
        boxrule=0pt, leftrule=3pt, arc=0pt, auto outer arc,
        left=8pt, right=8pt, top=5pt, bottom=5pt,
        before skip=8pt, after skip=8pt,
    },
}

\newcommand{\mymethod}{KL-MOPD\xspace}

\copyrightyear{2027}
\acmYear{2027}
\setcopyright{cc}
\setcctype{by-nc-nd}
\acmConference[arXiv 2026] {the arXiv Preprints}{October, 2026}{Hangzhou, China.}

\makeatletter
\newcommand{\authornotetext}[2]{%
  \if@ACM@anonymous\else
    \g@addto@macro\@authornotes{%
      \footnotetext[#1]{#2}}%
  \fi
}
\makeatother
\makeatletter
\let\ACMOriginalFnSymbol\@fnsymbol
\renewcommand*{\@fnsymbol}[1]{%
  \ifcase#1\relax
  \or
    \textnormal{\Letter}
  \or
    \ACMOriginalFnSymbol{2}
  \else
    \textnormal{\number\numexpr#1-2\relax}
  \fi
}
\makeatother

\begin{document}

\title[KL-MOPD: Discrepancy-Aware Multi-Objective On-Policy Distillation for Generative Recommendation]{
KL-MOPD: Discrepancy-Aware Multi-Objective On-Policy Distillation for Generative Recommendation
}


\settopmatter{authorsperrow=4}

\author{Weiqin Yang}
    \orcid{0000-0002-5750-5515}
    \affiliation{
        \institution{Zhejiang University}
        \city{Hangzhou}
        \country{China}
    }
    \email{tinysnow@zju.edu.cn}
\authornotemark[2]
\authornotemark[3]
\authornotemark[4]

\author{Yandong Wang}
    \orcid{0009-0003-6047-1728}
    \affiliation{
        \institution{Alibaba Group}
        \city{Beijing}
        \country{China}
    }
    \email{xvzhi.wyd@alibaba-inc.com}
\authornotemark[2]
\authornotemark[4]

\author{Yuan Gao}
    \orcid{0009-0004-2055-9941}
    \affiliation{
        \institution{Alibaba Group}
        \city{Beijing}
        \country{China}
    }
    \email{jianbei.gy@alibaba-inc.com}
\authornotemark[2]
\authornotemark[4]

\author{Hao Deng}
    \orcid{0009-0002-6335-7405}
    \affiliation{
        \institution{Alibaba Group}
        \city{Beijing}
        \country{China}
    }
    \email{denghao.deng@alibaba-inc.com}
\authornotemark[4]

\author{Jinxin Hu}
    \orcid{0000-0002-7252-5207}
    \affiliation{
        \institution{Alibaba Group}
        \city{Beijing}
        \country{China}
    }
    \email{jinxin.hjx@alibaba-inc.com}
\authornotemark[1]
\authornotemark[4]

\author{Jing Cao}
    \orcid{0000-0002-9485-4204}
    \affiliation{
        \institution{Zhejiang University}
        \city{Hangzhou}
        \country{China}
    }
    \email{cjcongjia@zju.edu.cn}

\author{Jiawei Chen}
    \orcid{0000-0002-4752-2629}
    \affiliation{
        \institution{Zhejiang University}
        \city{Hangzhou}
        \country{China}
    }
    \email{sleepyhunt@zju.edu.cn}
\authornotemark[1]
\authornotemark[3]

\authornotetext{2}{Equal contribution.}
\authornotetext{1}{Corresponding author.}
\authornotetext{3}{College of Computer Science and Technology,
  Zhejiang University.}
\authornotetext{4}{Alibaba International Digital Commerce Group.}

\ifanonymous
\else
    \renewcommand{\shortauthors}{Weiqin Yang et al.}
\fi



\begin{abstract}
    Coordinating multiple recommendation objectives (\eg click and conversion) is critical for real-world generative recommendation (GR). However, since a positive item in GR can be retrieved only when its tokens are successively selected at every decoding step, traditional item-wise methods lack fine-grained token-wise supervision to guide intermediate steps. Motivated by its success in foundation models, we consider introducing \emph{multi-objective on-policy distillation} (MOPD) into GR to distill two single-objective teachers into a student at each token. Yet, we find that in vanilla MOPD, (i) teachers exhibit unequal distillation difficulty at each token (\ie \emph{token-level imbalance}), and (ii) the harder-to-distill teacher varies across tokens (\ie \emph{cross-token inconsistency}), leaving each objective poorly learned at certain steps and hindering positive item retrieval for both objectives, leading to suboptimal performance.

    To address this challenge, we propose \textbf{\mymethod} (\textbf{\uline{KL}}-Weighted \textbf{\uline{M}}ulti-Objective \textbf{\uline{O}}n-\textbf{\uline{P}}olicy \textbf{\uline{D}}istillation), a simple yet effective distillation method that allocates token-wise teacher weights proportional to the initial student-teacher token-wise KL divergences, thereby strengthening supervision for the harder-to-distill teacher at every decoding step. Despite its simplicity, we theoretically prove that \mymethod strictly outperforms vanilla MOPD in the additive joint recommendation objective under mild conditions. Moreover, \mymethod is plug-and-play and computationally efficient, requiring only a one-time pre-training weight calibration with negligible training overhead. Extensive experiments on three public benchmarks, a billion-scale industrial dataset, and online A/B tests on our e-commerce advertising platform demonstrate that \mymethod significantly improves both individual and joint recommendation performance. Source code: \url{https://github.com/Yandorn/KL-MOPD}.
\end{abstract}

\ccsdesc[500]{Information systems~Recommender systems}

\keywords{Generative Recommendation; Multi-Objective Alignment; On-Policy Distillation}



\maketitle



\section{Introduction} \label{sec:introduction}

Generative recommendation (GR) has recently emerged as a transformative paradigm in recommender systems, garnering widespread attention across both academia and industry~\citep{rajput2023recommender,wang2024learnable,zheng2024adapting,deng2025onerec}. Unlike traditional recommenders that score and retrieve whole items, GR represents each item as a sequence of discrete semantic tokens and directly generates target items via token-by-token autoregressive decisions, unlocking remarkable scalability and generalization capabilities~\citep{rajput2023recommender,van2017neural,lee2022autoregressive}. Meanwhile, this multi-token generation mechanism imposes stricter demands on optimization: a positive item can be retrieved only when its tokens are successively selected at every decoding step, whereas failing to include the target token at any step prevents the item from being retrieved~\citep{yang2026bear,yu2026apao}.

\begin{figure*}[t]
    \centering
    \includegraphics[width=0.95\textwidth]{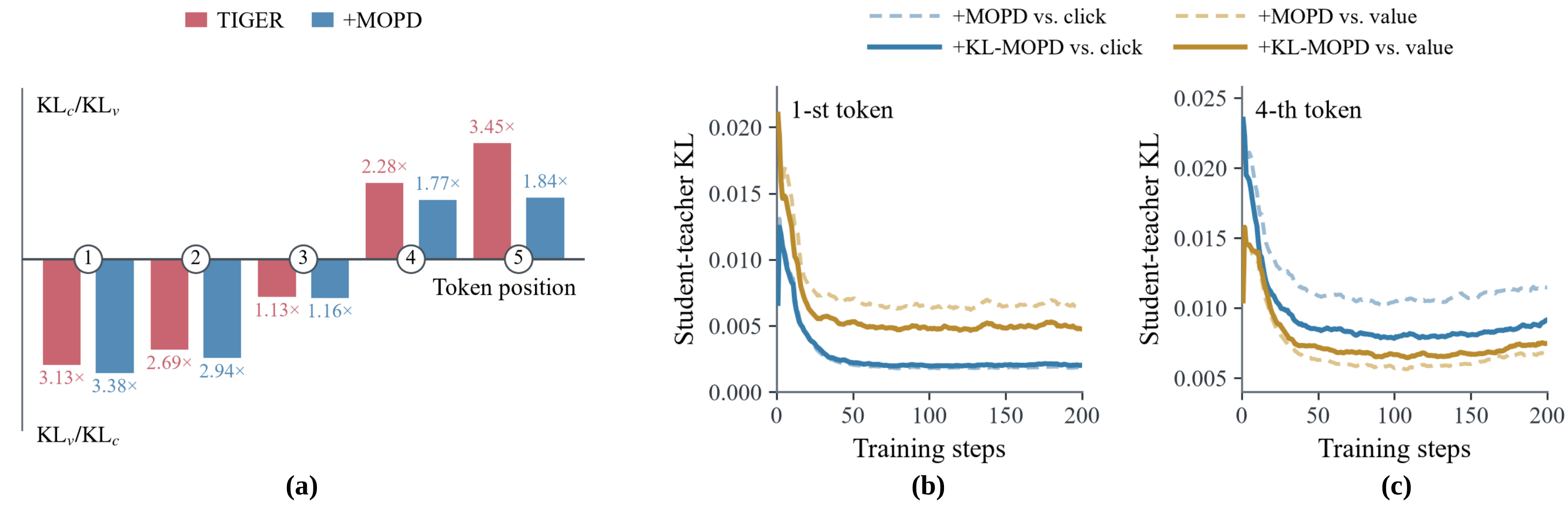}
    \caption{
        KL divergence between the student and teachers at each token using TIGER~\citep{rajput2023recommender} recommender on the TencentGR dataset~\citep{pan2026tencent}. Here, the click and value teachers are dedicated to the CTR and CVR objectives, respectively.
        \textbf{(a)} Ratios between the student-teacher KL divergences across tokens \#1 to \#5 for the initial student (TIGER) and the MOPD-trained student (+MOPD). Specifically, $\text{KL}_v / \text{KL}_c > 1$ indicates that the student is closer to the click teacher and is plotted below the axis, while $\text{KL}_c / \text{KL}_v > 1$ is plotted above.
        \textbf{(b, c)} Training curves of student-teacher KL at the 1-st and 4-th tokens, comparing +MOPD and our +\mymethod.
        Notably, MOPD inherits the token-level imbalance and cross-token shift of the initial student (\ie under-fitting the value teacher over the first three tokens and the click teacher over the last two tokens). In contrast, \mymethod effectively narrows the gap to the harder-to-distill teacher at each token, achieving more balanced distillation across decoding steps.
    }
    \Description{Illustration of the motivation behind KL-MOPD.}
    \label{fig:motivation}
\end{figure*}

\noindentparnoline{Multi-Objective GR.}
While existing studies on GR optimization primarily focus on a single recommendation objective, real-world recommender systems are typically required to balance multiple objectives simultaneously~\citep{xie2021personalized,huang2024data,tang2024multi}. For instance, in e-commerce scenarios, click-through rate (CTR) captures initial user engagement with recommended items, whereas conversion rate (CVR) reflects actual purchase conversions that drive transactional value~\citep{gu2020deep,zeng2021enhancing,meng2025generative}. However, these objectives do not always align: optimizing one in isolation far from guarantees a commensurate improvement in the other~\citep{lin2019pareto,pei2019value,weiss2025clicks}. Consequently, effectively coordinating these distinct objectives to improve joint recommendation performance remains a critical challenge for GR.

Over the past decade, extensive efforts have been devoted to multi-objective recommendation. Specifically, existing techniques coordinate multiple objectives either by directly balancing multi-objective losses and gradients~\citep{sener2018multi,yu2020gradient,liu2021conflict,lin2019pareto,pei2019value,xie2021personalized}, or by transferring and fusing model capabilities via item rankings~\citep{tang2018ranking,liu2020developing,kang2020rrd,yang2022cross,zhu2026rejuvenating}, hidden representations~\citep{kang2021topology,zhu2025exploring,ijcai2026p348}, and prediction scores~\citep{xv2022lightweight,tang2024multi,cao2025xmtf,he2025end}. However, being designed for traditional recommenders that score whole items, these methods operate only at the coarse \emph{item-wise} level. When adapted to GR, even if multiple objectives are aligned at the item level, the lack of fine-grained \emph{token-wise} supervision prevents the model from coordinating these objectives at intermediate decoding steps, which can still cause retrieval failures and lead to suboptimal joint recommendation performance~\citep{yang2026bear,yu2026apao}.

\noindentparnoline{Multi-Objective On-Policy Distillation.}
To overcome the limitations of traditional item-wise methods in GR, we consider introducing \emph{multi-objective on-policy distillation} (MOPD)~\citep{ma2026mopd,wang2026smopd}, a paradigm that has achieved remarkable success in merging specialized experts in foundation models. Following the standard multi-teacher distillation framework~\citep{liu2020developing,tang2024multi}, MOPD trains separate teacher models dedicated to individual objectives and distills their capabilities into a single student model. Specifically, it samples trajectories on-policy from the student and minimizes the average (\ie equally weighted) reverse Kullback-Leibler (KL) divergence from the student to each teacher's next-token probability distribution at every decoding step~\citep{agarwal2024policy,team2026onereason}. This mechanism makes MOPD particularly well-suited for GR due to two distinct advantages: (i) \emph{on-policy learning}, which optimizes within the student's own generation space and is equivalent to directly maximizing the KL-regularized joint recommendation objective (\cf \cref{sec:preliminaries:mopd}); and (ii) \emph{token-wise supervision}, which goes beyond coarse item-wise signals to provide fine-grained supervision at each token, naturally fitting the multi-token decoding process of GR.

\noindentparnoline{Challenges in MOPD.}
However, when directly applying vanilla MOPD, which assigns equal weights to both teachers at every decoding step, we make two key observations in GR:

\begin{itemize}[topsep=3pt,leftmargin=10pt,itemsep=3pt]
\item \emph{Teachers exhibit unequal distillation difficulty at each token.} At each token, we observe that the initial student is typically closer to one teacher while more distant from the other (\ie the \emph{harder-to-distill teacher}). For instance, at the first token of TIGER~\citep{rajput2023recommender} recommender on the TencentGR~\citep{pan2026tencent} dataset, the KL divergence from the initial student to the value teacher (for the CVR objective) is 3.13$\times$ that to the click teacher (for the CTR objective; \cf \cref{fig:motivation}(a)). Since vanilla MOPD treats both teachers equally regardless of this initial gap, the student continues to favor the closer teacher throughout training while leaving the harder-to-distill one severely under-fitted, resulting in \emph{imbalanced} capability transfer \emph{at each token} (\cf \cref{fig:motivation}(b, c)).

\item \emph{The harder-to-distill teacher varies across tokens.} Across decoding steps, which teacher is harder to distill does not remain fixed, but shifts between teachers along the token sequence. As \cref{fig:motivation}(a) exemplifies, the harder-to-distill teacher is the value teacher over the first three tokens, but becomes the click teacher over the last two tokens on the TencentGR dataset. Since vanilla MOPD applies static equal weights at all steps without adapting to this shift, the student continues to under-fit different teachers in the early or late decoding stages, resulting in \emph{inconsistent} capability transfer \emph{across tokens} (\cf \cref{fig:motivation}(a, b, c)).

\end{itemize}

\noindent
The above observations of \emph{token-level imbalance} and \emph{cross-token inconsistency} reveal a critical challenge for MOPD in GR. Specifically, for each objective, there always exist certain tokens where the student fails to adequately learn the corresponding teacher's decisions; since retrieving a positive item in GR requires selecting the appropriate token at every step, both objectives ultimately suffer from suboptimal performance (\cf \cref{tab:student-teacher-kl}).

\noindentparnoline{Our Method.}
To address this challenge, we propose \textbf{\mymethod}, a simple yet effective distillation method that allocates token-wise teacher supervision weights based on student-teacher discrepancies. Specifically, at each token, \mymethod assigns teacher weights in proportion to the token-wise KL divergence from the initial student to each teacher. In this way, \mymethod strengthens supervision for the harder-to-distill teacher (which is more distant from the initial student) at each token. This reduces the probability of decision errors for the poorly learned objective at each token, making the student more likely to select correct tokens at every step and retrieve positive items for both objectives. 

Despite its simplicity, \mymethod is backed by rigorous \emph{theoretical guarantees}. Specifically, for the canonical dual-objective setting~\citep{tang2024multi,yi2024dmmp}, we prove that \mymethod strictly outperforms vanilla MOPD in the additive joint recommendation objective under mild conditions, ensuring more effective multi-objective distillation (\cf \cref{theorem:kl-mopd-joint-objective}). Empirically, compared to vanilla MOPD, \mymethod achieves more balanced teacher fitting at every token (\cf \cref{fig:motivation}(b, c)) and further reduces the average item-wise KL divergence to both teachers by 13.70\%, consistently improving both individual and joint recommendation objectives (\cf \cref{tab:student-teacher-kl}).

In addition, \mymethod enjoys two appealing advantages for practical applications: (i) \emph{Plug-and-Play Design.} \mymethod can be seamlessly integrated into existing MOPD pipelines by simply reweighting the token-wise distillation loss, requiring neither auxiliary models nor architectural modifications. (ii) \emph{Computational Efficiency.} Since the teacher weights depend solely on the initial student and can be pre-computed only once prior to training, \mymethod introduces no extra forward passes during training and incurs negligible computational overhead. These properties make \mymethod a principled and practical solution for multi-objective GR.

We evaluate \mymethod on three public benchmark datasets using two advanced recommendation backbones against eight state-of-the-art baselines, most of which are newly proposed in 2026. Experimental results demonstrate that \mymethod significantly outperforms the strongest baseline on individual recommendation objectives (avg. +4.91\%), while surpassing both the strongest baseline and single-objective teachers on the joint recommendation objective (avg. +4.12\%). Furthermore, we validate \mymethod on a real-world e-commerce advertising platform: it improves the offline joint metric by 3.19\% on a billion-scale industrial dataset (over 1B interactions across 49M users and 158M ads), and achieves 1.63\% and 2.20\% uplifts in advertising revenue and gross merchandise volume (GMV), respectively, in \emph{online A/B tests with 10\% live traffic}.

\begin{statementbox}[title={Our Contributions}]
\begin{itemize}[topsep=0pt,leftmargin=3pt,itemsep=0pt]
\item We explore MOPD for multi-objective GR and uncover a critical challenge in vanilla MOPD from two key observations: (i) teachers exhibit unequal distillation difficulty at each token (\emph{token-level imbalance}), and (ii) the harder-to-distill teacher varies across tokens (\emph{cross-token inconsistency}), which together hinder positive item retrieval for both objectives.
\item We propose \mymethod, a plug-and-play and computationally efficient method that assigns token-wise teacher weights proportional to the initial student-teacher token-wise KL divergences, and theoretically prove that it strictly outperforms vanilla MOPD in the additive joint recommendation objective under mild conditions.
\item Extensive experiments on three public benchmarks, a billion-scale industrial dataset, and production online A/B tests demonstrate that \mymethod significantly improves both individual and joint recommendation performance.
\end{itemize}
\end{statementbox}


\section{Preliminaries} \label{sec:preliminaries}

In this section, we formulate the multi-objective generative recommendation task in \cref{sec:preliminaries:task-formulation}, introduce the multi-objective on-policy distillation (MOPD) paradigm and its potential for GR in \cref{sec:preliminaries:mopd}, and investigate the critical challenges of directly applying vanilla MOPD to GR in \cref{sec:preliminaries:challenges}.

\subsection{Task Formulation} \label{sec:preliminaries:task-formulation}

\noindentparnolinenospace{Generative Recommendation.}
In this work, we focus on generative recommendation (GR)~\citep{rajput2023recommender,wang2024learnable,zheng2024adapting,deng2025onerec}, where each item is represented as a sequence of semantic tokens $y = (y_1, y_2, \cdots, y_T)$. Here, $y_t$ denotes the $t$-th token and $T$ is the sequence length. Formally, given a user request $x \sim \mathcal{X}$ containing historical interactions and relevant context, a generative recommender with policy $\pi_\theta$ recommends a target item $y$ according to the conditional distribution $\pi_\theta(y | x)$, which is autoregressively factorized into token-wise probabilities $\pi_\theta(y | x) = \prod_{t=1}^{T} \pi_\theta(y_t | x, y_{<t})$, where $y_{<t} = (y_1, y_2, \cdots, y_{t-1})$ denotes the prefix tokens generated prior to step $t$. Due to this multiplicative factorization, successfully retrieving a positive item $y$ imposes strict demands on optimization: the policy $\pi_\theta(\cdot | x, y_{<t})$ must accurately select the appropriate token $y_t$ at every decoding step $t$, as a decision error at any step will deviate from the target prefix and lead to retrieval failure~\citep{yang2026bear,yu2026apao}.

\begin{table*}[t]
  \centering
  \caption{
    KL divergence between the student and teachers using the TIGER~\citep{rajput2023recommender} recommender on the TencentGR~\citep{pan2026tencent} and YooChoose~\citep{ben2015recsys} datasets.
    Here, TIGER denotes the initial student, while TIGER-C and TIGER-V represent the click and value teachers, respectively.
    The columns "$\text{KL}(y_t)$" ($t = 1, \cdots, 5$) and "$\text{KL}(y)$" report the token-wise and item-wise student-teacher KL divergences, respectively, and "$\Delta\text{KL}(y)$" indicates the relative reduction in $\text{KL}(y)$ compared to TIGER.
    As shown, \mymethod achieves more balanced token-wise KL divergence and further reduces the item-wise KL divergence \wrt both teachers.
  }
  \label{tab:student-teacher-kl}
  \resizebox{\textwidth}{!}{%
    \begin{tabular}{l|l|ccccc>{\columncolor[rgb]{0.93,0.96,1.00}}c>{\columncolor[rgb]{0.93,0.96,1.00}}c|ccccc>{\columncolor[rgb]{0.93,0.96,1.00}}c>{\columncolor[rgb]{0.93,0.96,1.00}}c}
    \toprule
    \multirow{2}{*}{\textbf{Methods}} & \multirow{2}{*}{\textbf{KL w.r.t.}} & \multicolumn{7}{c|}{\textbf{TencentGR}} & \multicolumn{7}{c}{\textbf{YooChoose}} \\
    \cmidrule(lr){3-9} \cmidrule(lr){10-16}
    & & $\bm{\mathrm{KL}(y_1)}$ & $\bm{\mathrm{KL}(y_2)}$ & $\bm{\mathrm{KL}(y_3)}$ & $\bm{\mathrm{KL}(y_4)}$ & $\bm{\mathrm{KL}(y_5)}$ & $\bm{\mathrm{KL}(y)}$ & $\bm{\Delta\mathrm{KL}(y)}$ & $\bm{\mathrm{KL}(y_1)}$ & $\bm{\mathrm{KL}(y_2)}$ & $\bm{\mathrm{KL}(y_3)}$ & $\bm{\mathrm{KL}(y_4)}$ & $\bm{\mathrm{KL}(y_5)}$ & $\bm{\mathrm{KL}(y)}$ & $\bm{\Delta\mathrm{KL}(y)}$ \\
    \midrule
    \multirow{2}{*}{TIGER} & TIGER-C & 0.00637 & 0.02301 & 0.03036 & 0.02319 & 0.05058 & 0.13351 & -- & 0.00860 & 0.01018 & 0.04381 & 0.03134 & 0.03326 & 0.12719 & -- \\
     & TIGER-V & 0.01995 & 0.06189 & 0.03437 & 0.01019 & 0.01467 & 0.14107 & -- & 0.02820 & 0.02511 & 0.02884 & 0.00998 & 0.00836 & 0.10049 & -- \\
    \cmidrule(lr){1-2} \cmidrule(lr){3-9} \cmidrule(lr){10-16}
    \multirow{2}{*}{+MOPD} & TIGER-C & 0.00191 & 0.00683 & 0.00981 & 0.01143 & 0.02920 & 0.05918 & -55.67\% & 0.00680 & 0.00593 & 0.01722 & 0.01390 & 0.01401 & 0.05786 & -54.51\% \\
     & TIGER-V & 0.00645 & 0.02006 & 0.01137 & 0.00645 & 0.01588 & 0.06021 & -57.32\% & 0.01385 & 0.01125 & 0.01203 & 0.00628 & 0.00507 & 0.04848 & -51.76\% \\
    \cmidrule(lr){1-2} \cmidrule(lr){3-9} \cmidrule(lr){10-16}
    \multirow{2}{*}{\textbf{+\mymethod}} & TIGER-C & 0.00227 & 0.00737 & 0.00788 & 0.00899 & 0.02421 & \textbf{0.05072} & \textbf{-62.01\%} & 0.00852 & 0.00683 & 0.01269 & 0.01002 & 0.00970 & \textbf{0.04776} & \textbf{-62.45\%} \\
     & TIGER-V & 0.00428 & 0.01333 & 0.00867 & 0.00701 & 0.01859 & \textbf{0.05188} & \textbf{-63.22\%} & 0.01034 & 0.00821 & 0.01136 & 0.00774 & 0.00679 & \textbf{0.04444} & \textbf{-55.78\%} \\
    \bottomrule
    \end{tabular}
  }
\end{table*}

\noindentparnoline{Multi-Objective Recommendation.}
Following the canonical dual-objective setting~\citep{tang2024multi,yi2024dmmp}, we consider two representative recommendation objectives, namely \emph{click} (C) and \emph{value} (V), characterized by objective functions $u_c(x, y)$ and $u_v(x, y)$, respectively.\footnote{We use click and value merely for ease of exposition; our method and theoretical analysis naturally apply to any pair of recommendation objectives.} Specifically, $u_c(x, y)$ measures user click engagement~\citep{zhou2018deep} with item $y$, whereas $u_v(x, y)$ captures downstream transactional value, such as conversion~\citep{ma2018entire}, GMV~\citep{meng2025generative}, or eCPM~\citep{xue2026generative}. Here, we assume that each objective function is non-negative and bounded, \ie $0 \leq u_j(x, y) \leq L_j$ for $j \in \{c, v\}$. For any policy $\pi$, its performance on objective $j \in \{c, v\}$ is evaluated by the expected objective value $\mathbb{E}_{x \sim \mathcal{X}, \, y \sim \pi(\cdot | x)}[u_j(x, y)]$. Ultimately, our goal is to learn a generative recommendation policy $\pi_\theta$ that maximizes the following conventional \emph{additive joint recommendation objective}:\footnote{Additive aggregation is widely adopted for combining multiple recommendation objectives~\citep{lin2019pareto,xie2021personalized,meng2025generative}. Other common formulations can also be easily converted into this additive form: for instance, a weighted sum $w_c u_c + w_v u_v$ for varying business priorities can be reduced by absorbing the weights (\ie $u_j \leftarrow w_j u_j$), while a multiplicative objective can be transformed into an additive one by taking the logarithm (\ie $u_j \leftarrow \log u_j$). Thus, we adopt the average form in \eqnref{eq:additive-joint-objective} without loss of generality.}
\begin{equation} \label{eq:additive-joint-objective}
    \max_\theta \; \mathbb{E}_{x \sim \mathcal{X}, y \sim \pi_\theta(\cdot | x)} \left[ u_c(x, y) + u_v(x, y) \right].
\end{equation}

\subsection{Multi-Objective On-Policy Distillation} \label{sec:preliminaries:mopd}

To optimize the joint recommendation objective in \eqnref{eq:additive-joint-objective}, existing multi-objective methods~\citep{lin2019pareto,xie2021personalized,xv2022lightweight,tang2024multi} are primarily designed for traditional recommenders and coordinate objectives exclusively at the coarse \emph{item-wise} level. When adapted to GR, even if multiple objectives are aligned at the item level, the lack of fine-grained \emph{token-wise} supervision leaves intermediate decoding steps unguided.
To overcome this limitation, we consider introducing \emph{multi-objective on-policy distillation} (MOPD)~\citep{ma2026mopd,wang2026smopd}, a paradigm that has achieved remarkable success in merging specialized experts in generative foundation models, into multi-objective GR.

\noindentparnoline{Single-Objective Teachers.}
To guide the student policy $\pi_\theta$ initialized from the reference policy $\pi_{\theta_0}$, MOPD leverages two pre-trained single-objective teachers, $\pi_c$ and $\pi_v$, dedicated to the click and value objectives, respectively. Following standard practice~\citep{rafailov2023direct}, we assume that each teacher $j \in \{c, v\}$ optimizes the following KL-regularized recommendation objective maximization problem:
\begin{equation} \label{eq:kl-regularized-teacher-objective}
    \max_{\pi_j} \; \mathbb{E}_{y \sim \pi_j(\cdot | x)}[u_j(x, y)] - \beta \mathbb{D}_{\mathrm{KL}} \left( \pi_j(\cdot | x) \,\|\, \pi_{\theta_0}(\cdot | x) \right),
\end{equation}
where $\beta > 0$ is a coefficient shared by both teachers that controls the KL penalty \wrt the reference policy $\pi_{\theta_0}$. According to classical derivations~\citep{peters2007reinforcement,rafailov2023direct}, the optimization problem in \eqnref{eq:kl-regularized-teacher-objective} admits a unique closed-form optimal teacher policy:
\begin{equation} \label{eq:optimal-teacher-policy}
    \pi_j(y | x) \propto \pi_{\theta_0}(y | x) \exp\left( \frac{u_j(x, y)}{\beta} \right).
\end{equation}

\noindentparnoline{Multi-Objective On-Policy Distillation.}
During MOPD, both teachers $\pi_c$ and $\pi_v$ remain fixed, while the student $\pi_\theta$ is initialized from the reference policy $\pi_{\theta_0}$ and updated on-policy. Specifically, for each request $x \sim \mathcal{X}$, the student first samples a target item sequence $y \sim \pi_\theta(\cdot | x)$ using its current policy. Conditioned on each student-generated prefix $(x, y_{<t})$, the two teachers then provide their next-token probability distributions $\pi_j(\cdot | x, y_{<t})$ ($j \in \{c, v\}$) to guide the student's decisions. Adopting the reverse KL divergence~\citep{agarwal2024policy,ma2026mopd} with equal teacher weights ($1/2$) at every decoding step, vanilla MOPD minimizes the item-wise KL divergence, which can be decomposed into token-wise objectives via the chain rule:
\begin{equation} \label{eq:vanilla-mopd-loss}
\begin{aligned}
    & \mathcal{L}_{\text{MOPD}}(\theta)
    = \mathbb{E}_{x \sim \mathcal{X}} \left[ \frac{1}{2} \sum_{j \in \{c, v\}} \mathbb{D}_{\mathrm{KL}} \left( \pi_\theta(\cdot | x) \,\|\, \pi_j(\cdot | x) \right) \right] \\
    & = \mathbb{E}_{\substack{x \sim \mathcal{X} \\ y \sim \pi_\theta(\cdot | x)}} \left[ \sum_{t=1}^{T} \frac{1}{2} \sum_{j \in \{c, v\}} \mathbb{D}_{\mathrm{KL}} \left( \pi_\theta(\cdot | x, y_{<t}) \,\|\, \pi_j(\cdot | x, y_{<t}) \right) \right].
\end{aligned}
\end{equation}

\noindentparnoline{Discussions: Why MOPD for GR?}
Compared to traditional multi-objective recommendation methods, MOPD in \eqnref{eq:vanilla-mopd-loss} is particularly well-suited for generative recommendation in two aspects:

\begin{itemize}[topsep=3pt,leftmargin=10pt,itemsep=3pt]

\item \emph{On-Policy Learning.}
MOPD optimizes the student policy $\pi_\theta$ within its own generation space and is equivalent to directly maximizing the KL-regularized additive joint objective (\cf \cref{eq:additive-joint-objective,eq:kl-regularized-teacher-objective}). Specifically, by substituting the optimal teacher policies from \eqnref{eq:optimal-teacher-policy} into \eqnref{eq:vanilla-mopd-loss}, we can rewrite the MOPD loss up to a $\theta$-independent constant $C$ as:
\begin{equation} \label{eq:mopd-joint-equivalence}
\begin{aligned}
    2\beta \mathcal{L}_{\text{MOPD}}(\theta) 
    & = -\mathbb{E}_{\substack{x \sim \mathcal{X} \\ y \sim \pi_\theta(\cdot | x)}} \left[ u_c(x, y) + u_v(x, y) \right] \\
    & \quad + 2\beta \mathbb{E}_{x \sim \mathcal{X}} \left[ \mathbb{D}_{\mathrm{KL}} \left( \pi_\theta(\cdot | x) \,\|\, \pi_{\theta_0}(\cdot | x) \right) \right] + C.
\end{aligned}
\end{equation}
This equivalence reveals why on-policy learning is essential: evaluating teacher feedback on the student's self-generated trajectories $y \sim \pi_\theta(\cdot | x)$ directly maximizes its expected joint objective $u_c(x, \cdot) + u_v(x, \cdot)$~\citep{agarwal2024policy}, whereas off-policy methods operate in static corpus or teacher rollout spaces, suffering from a distribution gap that deviates from the recommendation objective (\eqnref{eq:additive-joint-objective}) and leads to suboptimal performance.

\item \emph{Token-Wise Supervision.}
MOPD directly constrains the next-token probability distribution at every decoding step, naturally fitting the multi-token generation process of GR. Specifically, traditional item-wise methods merely align the overall item-level probability $\pi_\theta(y | x) = \prod_{t=1}^{T} \pi_\theta(y_t | x, y_{<t})$ without explicitly constraining each individual token probability $\pi_\theta(y_t | x, y_{<t})$. Consequently, even if the overall item probability is aligned, a low probability at any intermediate token $y_t$ can still cause a decision error during decoding and lead to retrieval failure~\citep{yang2026bear,yu2026apao}. In contrast, through the chain-rule decomposition in \eqnref{eq:vanilla-mopd-loss}, MOPD directly aligns the student's token-wise distribution $\pi_\theta(\cdot | x, y_{<t})$ with each teacher's $\pi_j(\cdot | x, y_{<t})$ at every step $t$, providing fine-grained supervision to mitigate intermediate decoding errors~\citep{yang2026bear}.

\end{itemize}

\subsection{The Pitfalls of MOPD in GR} \label{sec:preliminaries:challenges}

Although vanilla MOPD provides fine-grained token-wise supervision in \eqnref{eq:vanilla-mopd-loss}, it assigns a static, equal weight ($1/2$) to both teachers across all decoding steps. To examine its actual behavior in multi-objective GR, we conduct an empirical study (\cf \cref{fig:motivation,tab:student-teacher-kl}) by measuring the expected token-wise KL divergence $\mathbb{E}_{x \sim \mathcal{X}, \, y_{<t} \sim \pi_\theta(\cdot | x)} \left[ \mathbb{D}_{\mathrm{KL}} \left( \pi_\theta(\cdot | x, y_{<t}) \,\|\, \pi_j(\cdot | x, y_{<t}) \right) \right]$ from the initial student $\pi_{\theta_0}$ and the trained student $\pi_\theta$ to each teacher $\pi_j$ ($j \in \{c, v\}$) at each step $t$. Our results reveal two key observations:

\noindentparnoline{Observation 1: Teachers Exhibit Unequal Distillation Difficulty at Each Token.}
At each token $t$, we observe that the initial student $\pi_{\theta_0}$ is typically closer to one teacher while more distant from the other, which we refer to as the \emph{harder-to-distill teacher}. For instance, at the first token ($t = 1$) of TIGER~\citep{rajput2023recommender} on the TencentGR~\citep{pan2026tencent} dataset, the KL divergence from the initial student $\pi_{\theta_0}$ to the value teacher $\pi_v$ is $3.13\times$ that to the click teacher $\pi_c$ (\cf \cref{fig:motivation}(a) and \cref{tab:student-teacher-kl}). Since vanilla MOPD assigns equal weights to both teachers \emph{regardless of this initial gap}, the student $\pi_\theta$ continues to favor the closer teacher throughout training while leaving the harder-to-distill one severely under-fitted (\cf \cref{fig:motivation}(b) and \cref{tab:student-teacher-kl}). After training, the KL divergence from the student $\pi_\theta$ to the value teacher $\pi_v$ at the first token still remains $3.38\times$ that to the click teacher $\pi_c$ (\cf \cref{fig:motivation}(a)), resulting in \emph{imbalanced capability transfer at each token} (\ie \emph{\textbf{token-level imbalance}}).

\noindentparnoline{Observation 2: The Harder-to-Distill Teacher Varies Across Tokens.}
Across decoding steps, which teacher is harder to distill does not remain fixed, but shifts between teachers along the token sequence. As \cref{fig:motivation}(a) exemplifies, the harder-to-distill teacher for the initial student $\pi_{\theta_0}$ is the value teacher $\pi_v$ over the first three tokens ($t = 1, 2, 3$), but becomes the click teacher $\pi_c$ over the last two tokens ($t = 4, 5$). Since vanilla MOPD applies static equal weights at all steps \emph{without adapting to this shift}, the student $\pi_\theta$ continues to under-fit different teachers in the early or late decoding stages throughout training (\cf \cref{fig:motivation}(b, c)). After training, the value teacher $\pi_v$ remains under-distilled over the first three tokens while the click teacher $\pi_c$ remains under-distilled over the last two tokens (\cf \cref{tab:student-teacher-kl}), resulting in \emph{inconsistent capability transfer across tokens} (\ie \emph{\textbf{cross-token inconsistency}}).

\begin{table}[t]
  \centering
  \caption{Statistics of the public and industrial datasets.}
  \label{tab:dataset-statistics}
  \small
  \setlength{\tabcolsep}{3pt}
  \renewcommand{\arraystretch}{1.12}
  \begin{tabular*}{\columnwidth}{@{\extracolsep{\fill}}l|rrrr}
    \toprule
    \textbf{Dataset} & \textbf{\#Users} & \textbf{\#Items} & \textbf{\#Interactions} & \textbf{Density} \\
    \midrule
    TencentGR & 49,803 & 1,462,146 & 4,755,945 & $6.5\times 10^{-5}$ \\
    YooChoose & 409,981 & 38,974 & 5,074,843 & $3.2\times 10^{-4}$ \\
    Taobao & 14,298 & 413,730 & 928,509 & $1.6\times 10^{-4}$ \\
    Industry & 49,277,990 & 158,533,241 & 1,492,770,272 & $1.9\times 10^{-7}$ \\
    \bottomrule
  \end{tabular*}
\end{table}

\noindentparnoline{Resulting Pitfalls.}
Combined, the above observations of \emph{token-level imbalance} and \emph{cross-token inconsistency} expose a critical bottleneck of vanilla MOPD in GR. Specifically, for either objective $j \in \{c, v\}$, there always exist certain decoding steps $t$ where the student distribution $\pi_\theta(\cdot | x, y_{<t})$ fails to adequately align with the teacher $\pi_j(\cdot | x, y_{<t})$. Under the autoregressive probability factorization $\pi_\theta(y | x) = \prod_{t=1}^{T} \pi_\theta(y_t | x, y_{<t})$, poor alignment at those steps makes the student prone to token decision errors, which derails the generation trajectory and ultimately compromises positive item retrieval for both objectives (\cf \cref{tab:student-teacher-kl}).

\begin{statementbox}[title={Challenge: Pitfalls of Vanilla MOPD in GR}]
By overlooking the varying distillation difficulties induced by dynamic student-teacher discrepancies, vanilla MOPD suffers from \emph{token-level imbalance} and \emph{cross-token inconsistency} distillation, leaving each objective poorly learned at certain tokens and hindering positive item retrieval for both objectives.
\end{statementbox}



\begin{table*}[t]
  \centering
  \caption{
    Overall performance comparison on TIGER~\citep{rajput2023recommender} and REG4Rec~\citep{xing2026reg4rec} backbones. The student is initialized with the pre-trained backbone and distilled from two teachers (\ie "+" in the method name). The columns "$u_c$" and "$u_v$" denote the individual objectives for the click and value teachers (\ie Backbone-C and Backbone-V), respectively, while the light-blue columns report their average $u_{cv} = (u_c + u_v)/2$ as the additive joint objective. The best results are highlighted in bold, and the second-best results are underlined. "\textcolor{darkred}{\textbf{Imp.\%}}" indicates the relative improvement of \mymethod over the strongest baseline, where $^{*}$ denotes that for the joint objective $u_{cv}$, the improvement is computed over both baselines and single-objective teachers.
  }
  \label{tab:main-results}
  \resizebox{\textwidth}{!}{%
    \begin{tabular}{l|cc>{\columncolor[rgb]{0.93,0.96,1.00}}ccc>{\columncolor[rgb]{0.93,0.96,1.00}}c|cc>{\columncolor[rgb]{0.93,0.96,1.00}}ccc>{\columncolor[rgb]{0.93,0.96,1.00}}c|cc>{\columncolor[rgb]{0.93,0.96,1.00}}ccc>{\columncolor[rgb]{0.93,0.96,1.00}}c}
    \toprule
    \multicolumn{1}{c|}{\multirow{2}[4]{*}{\textbf{Methods}}} & \multicolumn{6}{c|}{\textbf{TencentGR}} & \multicolumn{6}{c|}{\textbf{YooChoose}} & \multicolumn{6}{c}{\textbf{Taobao}} \\
    \cmidrule(lr){2-7} \cmidrule(lr){8-13} \cmidrule(lr){14-19}
    & \textbf{\boldmath$u_{c}$@10} & \textbf{\boldmath$u_{v}$@10} & \textbf{\boldmath$u_{cv}$@10} & \textbf{\boldmath$u_{c}$@20} & \textbf{\boldmath$u_{v}$@20} & \textbf{\boldmath$u_{cv}$@20} & \textbf{\boldmath$u_{c}$@10} & \textbf{\boldmath$u_{v}$@10} & \textbf{\boldmath$u_{cv}$@10} & \textbf{\boldmath$u_{c}$@20} & \textbf{\boldmath$u_{v}$@20} & \textbf{\boldmath$u_{cv}$@20} & \textbf{\boldmath$u_{c}$@10} & \textbf{\boldmath$u_{v}$@10} & \textbf{\boldmath$u_{cv}$@10} & \textbf{\boldmath$u_{c}$@20} & \textbf{\boldmath$u_{v}$@20} & \textbf{\boldmath$u_{cv}$@20} \\
    \midrule
    \textbf{TIGER} & 0.0108 & 0.0082 & 0.0095 & 0.0191 & 0.0158 & 0.0175 & 0.4115 & 0.5016 & 0.4566 & 0.5214 & 0.6034 & 0.5624 & 0.0434 & 0.1325 & 0.0880 & 0.0574 & 0.1672 & 0.1123 \\
    TIGER-C & \uline{0.0286} & 0.0164 & \uline{0.0225} & \textbf{0.0440} & 0.0265 & \uline{0.0353} & \textbf{0.4359} & 0.5193 & 0.4776 & \textbf{0.5447} & 0.6188 & 0.5818 & \uline{0.0665} & 0.1963 & 0.1314 & \uline{0.0854} & 0.2500 & 0.1677 \\
    TIGER-V & 0.0209 & \textbf{0.0189} & 0.0199 & 0.0321 & \textbf{0.0334} & 0.0328 & 0.3810 & \textbf{0.5849} & 0.4830 & 0.4934 & \textbf{0.6889} & \uline{0.5912} & 0.0627 & \uline{0.2006} & \uline{0.1317} & 0.0818 & \uline{0.2561} & \uline{0.1690} \\
    \cmidrule(lr){1-1} \cmidrule(lr){2-7} \cmidrule(lr){8-13} \cmidrule(lr){14-19}
    +RCE-KD & 0.0258 & 0.0173 & 0.0216 & 0.0377 & 0.0290 & 0.0334 & 0.4125 & 0.5216 & 0.4671 & 0.5326 & 0.6351 & 0.5839 & 0.0605 & 0.1839 & 0.1222 & 0.0769 & 0.2228 & 0.1499 \\
    +MTS$^2$4MM & 0.0232 & 0.0154 & 0.0193 & 0.0356 & 0.0280 & 0.0318 & 0.4190 & 0.5276 & 0.4733 & 0.5291 & 0.6455 & 0.5873 & 0.0637 & 0.1937 & 0.1287 & 0.0805 & 0.2411 & 0.1608 \\
    +SmartGR & 0.0253 & 0.0155 & 0.0204 & 0.0381 & 0.0284 & 0.0333 & 0.4134 & 0.5291 & 0.4713 & 0.5231 & 0.6317 & 0.5774 & 0.0629 & 0.1873 & 0.1251 & 0.0813 & 0.2367 & 0.1590 \\
    +GR4AD & 0.0241 & 0.0168 & 0.0205 & 0.0381 & 0.0268 & 0.0325 & 0.4229 & 0.5226 & 0.4728 & 0.5328 & 0.6233 & 0.5781 & 0.0637 & 0.1918 & 0.1278 & 0.0817 & 0.2452 & 0.1635 \\
    +TSGR & 0.0247 & 0.0165 & 0.0206 & 0.0381 & 0.0274 & 0.0328 & 0.4130 & 0.5231 & 0.4681 & 0.5328 & 0.6238 & 0.5783 & 0.0636 & 0.1924 & 0.1280 & 0.0817 & 0.2442 & 0.1630 \\
    +MOPD & 0.0261 & 0.0164 & 0.0213 & 0.0380 & 0.0280 & 0.0330 & 0.4186 & 0.5265 & 0.4726 & 0.5301 & 0.6309 & 0.5805 & 0.0652 & 0.1941 & 0.1297 & 0.0844 & 0.2517 & 0.1681 \\
    +SMOPD & 0.0253 & 0.0170 & 0.0212 & 0.0384 & 0.0280 & 0.0332 & 0.4179 & 0.5517 & \uline{0.4848} & 0.5271 & 0.6549 & 0.5910 & 0.0635 & 0.1931 & 0.1283 & 0.0836 & 0.2503 & 0.1670 \\
    +TrustMOPD & 0.0256 & 0.0164 & 0.0210 & 0.0375 & 0.0284 & 0.0330 & 0.4046 & 0.5492 & 0.4769 & 0.5153 & 0.6573 & 0.5863 & 0.0616 & 0.1883 & 0.1250 & 0.0802 & 0.2452 & 0.1627 \\
    \cmidrule(lr){1-1} \cmidrule(lr){2-7} \cmidrule(lr){8-13} \cmidrule(lr){14-19}
    \textbf{+\mymethod} & \textbf{0.0294} & \uline{0.0186} & \textbf{0.0240}$^{*}$ & \uline{0.0417} & \uline{0.0312} & \textbf{0.0365}$^{*}$ & \uline{0.4327} & \uline{0.5741} & \textbf{0.5034}$^{*}$ & \uline{0.5411} & \uline{0.6762} & \textbf{0.6087}$^{*}$ & \textbf{0.0668} & \textbf{0.2044} & \textbf{0.1356}$^{*}$ & \textbf{0.0865} & \textbf{0.2623} & \textbf{0.1744}$^{*}$ \\
    \textcolor{darkred}{\textbf{Imp.\%}} & \textcolor{darkred}{\textbf{+12.64\%}} & \textcolor{darkred}{\textbf{+7.51\%}} & \textcolor{darkred}{\textbf{+6.67\%}$^{*}$} & \textcolor{darkred}{\textbf{+8.59\%}} & \textcolor{darkred}{\textbf{+7.59\%}} & \textcolor{darkred}{\textbf{+3.40\%}$^{*}$} & \textcolor{darkred}{\textbf{+2.32\%}} & \textcolor{darkred}{\textbf{+4.06\%}} & \textcolor{darkred}{\textbf{+3.84\%}$^{*}$} & \textcolor{darkred}{\textbf{+1.56\%}} & \textcolor{darkred}{\textbf{+2.88\%}} & \textcolor{darkred}{\textbf{+2.96\%}$^{*}$} & \textcolor{darkred}{\textbf{+2.45\%}} & \textcolor{darkred}{\textbf{+5.31\%}} & \textcolor{darkred}{\textbf{+3.00\%}$^{*}$} & \textcolor{darkred}{\textbf{+2.49\%}} & \textcolor{darkred}{\textbf{+4.21\%}} & \textcolor{darkred}{\textbf{+3.23\%}$^{*}$} \\
    \midrule
    \textbf{REG4Rec} & 0.0285 & 0.0240 & 0.0263 & 0.0443 & 0.0388 & 0.0416 & 0.4091 & 0.5218 & 0.4655 & 0.5178 & 0.6266 & 0.5722 & 0.0671 & 0.1948 & 0.1310 & 0.0831 & 0.2302 & 0.1567 \\
    REG4Rec-C & \uline{0.0381} & 0.0262 & \uline{0.0322} & \uline{0.0547} & 0.0438 & \uline{0.0493} & \textbf{0.4238} & 0.5414 & 0.4826 & \textbf{0.5327} & 0.6427 & 0.5877 & 0.0754 & 0.2228 & 0.1491 & \uline{0.0933} & 0.2643 & 0.1788 \\
    REG4Rec-V & 0.0291 & \uline{0.0293} & 0.0292 & 0.0463 & \uline{0.0448} & 0.0456 & 0.3790 & \textbf{0.5862} & 0.4826 & 0.4917 & \textbf{0.6942} & \uline{0.5930} & 0.0743 & 0.2241 & 0.1492 & 0.0920 & 0.2681 & 0.1801 \\
    \cmidrule(lr){1-1} \cmidrule(lr){2-7} \cmidrule(lr){8-13} \cmidrule(lr){14-19}
    +RCE-KD & 0.0328 & 0.0277 & 0.0303 & 0.0523 & 0.0429 & 0.0476 & 0.4103 & 0.5401 & 0.4752 & 0.5192 & 0.6445 & 0.5819 & 0.0742 & 0.2166 & 0.1454 & 0.0904 & 0.2585 & 0.1745 \\
    +MTS$^2$4MM & 0.0315 & 0.0262 & 0.0289 & 0.0494 & 0.0422 & 0.0458 & 0.4134 & 0.5385 & 0.4760 & 0.5152 & 0.6380 & 0.5766 & 0.0637 & 0.1836 & 0.1237 & 0.0802 & 0.2241 & 0.1522 \\
    +SmartGR & 0.0330 & 0.0277 & 0.0304 & 0.0511 & 0.0419 & 0.0465 & 0.4112 & 0.5390 & 0.4751 & 0.5205 & 0.6431 & 0.5818 & 0.0678 & 0.1972 & 0.1325 & 0.0844 & 0.2394 & 0.1619 \\
    +GR4AD & 0.0304 & 0.0249 & 0.0277 & 0.0490 & 0.0425 & 0.0458 & 0.4127 & 0.5349 & 0.4738 & 0.5212 & 0.6474 & 0.5843 & 0.0636 & 0.1836 & 0.1236 & 0.0800 & 0.2245 & 0.1523 \\
    +TSGR & 0.0303 & 0.0249 & 0.0276 & 0.0478 & 0.0413 & 0.0445 & 0.4126 & 0.5355 & 0.4741 & 0.5206 & 0.6478 & 0.5842 & 0.0631 & 0.1812 & 0.1221 & 0.0790 & 0.2204 & 0.1497 \\
    +MOPD & 0.0338 & 0.0277 & 0.0308 & 0.0505 & 0.0422 & 0.0464 & 0.4107 & 0.5386 & 0.4747 & 0.5203 & 0.6450 & 0.5827 & 0.0744 & \uline{0.2282} & \uline{0.1513} & 0.0923 & 0.2674 & 0.1799 \\
    +SMOPD & 0.0344 & 0.0277 & 0.0311 & 0.0512 & 0.0436 & 0.0474 & 0.4105 & 0.5617 & \uline{0.4861} & 0.5204 & 0.6634 & 0.5919 & \uline{0.0761} & 0.2245 & 0.1503 & \uline{0.0933} & 0.2650 & 0.1792 \\
    +TrustMOPD & 0.0350 & 0.0284 & 0.0317 & 0.0518 & 0.0438 & 0.0478 & 0.3957 & 0.5642 & 0.4800 & 0.5062 & 0.6671 & 0.5867 & 0.0750 & 0.2241 & 0.1496 & 0.0932 & \uline{0.2708} & \uline{0.1820} \\
    \cmidrule(lr){1-1} \cmidrule(lr){2-7} \cmidrule(lr){8-13} \cmidrule(lr){14-19}
    \textbf{+\mymethod} & \textbf{0.0384} & \textbf{0.0303} & \textbf{0.0344}$^{*}$ & \textbf{0.0573} & \textbf{0.0451} & \textbf{0.0512}$^{*}$ & \uline{0.4221} & \uline{0.5855} & \textbf{0.5038}$^{*}$ & \uline{0.5322} & \uline{0.6902} & \textbf{0.6112}$^{*}$ & \textbf{0.0784} & \textbf{0.2401} & \textbf{0.1593}$^{*}$ & \textbf{0.0975} & \textbf{0.2793} & \textbf{0.1884}$^{*}$ \\
    \textcolor{darkred}{\textbf{Imp.\%}} & \textcolor{darkred}{\textbf{+9.71\%}} & \textcolor{darkred}{\textbf{+6.69\%}} & \textcolor{darkred}{\textbf{+6.84\%}$^{*}$} & \textcolor{darkred}{\textbf{+9.56\%}} & \textcolor{darkred}{\textbf{+2.97\%}} & \textcolor{darkred}{\textbf{+3.96\%}$^{*}$} & \textcolor{darkred}{\textbf{+2.10\%}} & \textcolor{darkred}{\textbf{+3.78\%}} & \textcolor{darkred}{\textbf{+3.64\%}$^{*}$} & \textcolor{darkred}{\textbf{+2.11\%}} & \textcolor{darkred}{\textbf{+3.46\%}} & \textcolor{darkred}{\textbf{+3.08\%}$^{*}$} & \textcolor{darkred}{\textbf{+3.02\%}} & \textcolor{darkred}{\textbf{+5.21\%}} & \textcolor{darkred}{\textbf{+5.25\%}$^{*}$} & \textcolor{darkred}{\textbf{+4.50\%}} & \textcolor{darkred}{\textbf{+3.14\%}} & \textcolor{darkred}{\textbf{+3.52\%}$^{*}$} \\
    \bottomrule
    \end{tabular}
  }
\end{table*}

\section{Methodology} \label{sec:methodology}

As revealed in \cref{sec:preliminaries:challenges}, vanilla MOPD is insufficient to handle the complex distillation dynamics in GR (\ie token-level imbalance and cross-token inconsistency). Motivated by this, we propose \textbf{\mymethod} (\textbf{\uline{KL}}-Weighted \textbf{\uline{M}}ulti-Objective \textbf{\uline{O}}n-\textbf{\uline{P}}olicy \textbf{\uline{D}}istillation), a simple yet principled method that leverages token-wise weighting based on student-teacher discrepancies to dynamically adapt to the varying distillation difficulties of different teachers across tokens. In the following, we first present the details of \mymethod in \cref{sec:methodology:kl-mopd}, and then analyze its theoretical guarantees and practical advantages compared to existing methods in \cref{sec:methodology:analysis}.

\subsection{KL-Weighted On-Policy Distillation} \label{sec:methodology:kl-mopd}

\noindentparnolinenospace{Measuring Distillation Difficulties by Discrepancies.}
Vanilla MOPD suffers from token-level imbalance and cross-token inconsistency in GR, where the student under-distills different harder-to-distill teachers across tokens (\cf \cref{sec:preliminaries:challenges}). Consequently, both recommendation objectives remain poorly learned at certain decoding steps, which derails multi-token generation and degrades joint recommendation performance. To address this issue, the primary step is to quantify the distillation difficulty at each token so as to identify which teacher is harder to distill. As evidenced by our empirical observations in \cref{fig:motivation,tab:student-teacher-kl}, a larger discrepancy between the initial student $\pi_{\theta_0}$ and a teacher $\pi_j$, as reflected by a larger KL divergence, consistently makes $\pi_j$ more challenging to distill in MOPD. Motivated by this observation, we measure the distillation difficulty of each teacher $\pi_j$ ($j \in \{c, v\}$) at each token $t$ via the expected token-wise KL divergence from $\pi_{\theta_0}$ to $\pi_j$:
\begin{equation} \label{eq:initial-kl}
    d_{j, t} \coloneqq \mathbb{E}_{\substack{x \sim \mathcal{X} \\ y \sim \pi_{\theta_0}(\cdot | x)}} \left[ \mathbb{D}_{\mathrm{KL}} \left( \pi_{\theta_0}(\cdot | x, y_{<t}) \,\|\, \pi_j(\cdot | x, y_{<t}) \right) \right].
\end{equation}

\noindentparnoline{KL-Weighted On-Policy Distillation (KL-MOPD).}
To learn from different teachers in a more balanced manner at each token, a natural idea is to weight the teachers according to their distillation difficulties, \ie assigning a larger weight to the harder-to-distill teacher. Building on this idea, we propose \textbf{\mymethod}. Specifically, utilizing the distillation difficulty $d_{j, t}$ defined via the student-teacher token-wise KL divergence in \eqnref{eq:initial-kl}, we define the normalized weight for each teacher $j \in \{c, v\}$ at token $t$ as:\footnote{In practical distillation scenarios, we can assume $d_{j, t} > 0$ for $j \in \{c, v\}$.}
\begin{equation} \label{eq:weight}
    \alpha_{j, t} \coloneqq \frac{d_{j, t}}{d_{c, t} + d_{v, t}}, \quad j \in \{c, v\}.
\end{equation}
Replacing the equal weights ($1/2$) in \eqnref{eq:vanilla-mopd-loss} with our token-wise weights $\alpha_{j, t}$ in \eqnref{eq:weight}, \mymethod optimizes the following objective:
\begin{equation} \label{eq:kl-mopd-loss}
\begin{aligned}
    &\mathcal{L}_{\text{\mymethod}}(\theta) \coloneqq\\
    & \mathbb{E}_{\substack{x \sim \mathcal{X} \\ y \sim \pi_\theta(\cdot | x)}} \left[ \sum_{t=1}^{T} \sum_{j \in \{c, v\}} \alpha_{j, t} \mathbb{D}_{\mathrm{KL}} \left( \pi_\theta(\cdot | x, y_{<t}) \,\|\, \pi_j(\cdot | x, y_{<t}) \right) \right].
\end{aligned}
\end{equation}
Compared to vanilla MOPD, \mymethod assigns stronger supervision to the harder-to-distill teacher at every decoding step. This design (i) facilitates better fitting of the objective that was previously poorly learned at each token, mitigating \emph{token-level imbalance} (\cf \cref{fig:motivation}); and (ii) balances the learning progress of each individual objective across different tokens, alleviating \emph{cross-token inconsistency} (\cf \cref{tab:student-teacher-kl}). Ultimately, this makes the student less likely to suffer from item-level decoding failures for a target objective caused by under-distilling its teacher at certain tokens, thereby improving joint recommendation performance (\cf \cref{tab:main-results}).

\begin{figure*}[t]
  \centering
  \begin{minipage}[t]{\dimexpr(\textwidth-\columnsep)/2\relax}
    \vspace{0pt}
    \centering
    \includegraphics[width=\linewidth]{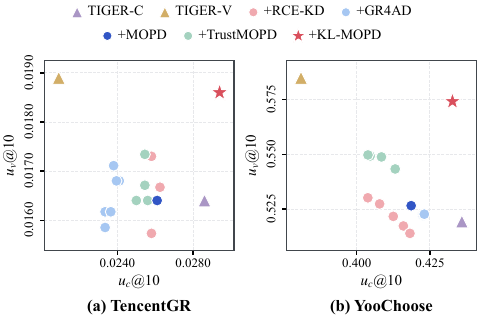}
    \caption{
      Pareto frontiers of click and value objectives for
      single-objective teachers, representative baselines,
      and \mymethod on TencentGR and YooChoose with TIGER backbone.
    }
    \Description{Pareto frontiers of click and value objectives
      at top-10 recommendation.}
    \label{fig:pareto_click_value}
  \end{minipage}\hfill%
  \begin{minipage}[t]{\dimexpr(\textwidth-\columnsep)/2\relax}
    \vspace{0pt}
    \centering
    \includegraphics[width=\linewidth]{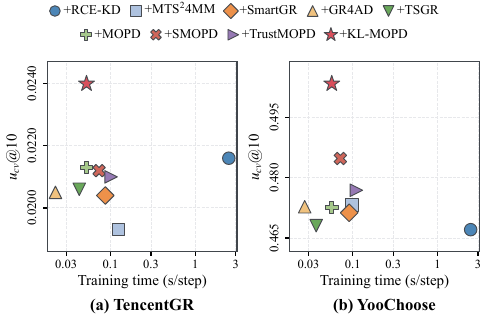}
    \caption{
      Pareto frontiers of the additive joint objective $u_{cv}$
      and training time (s/step) for all baselines and \mymethod
      on TencentGR and YooChoose with TIGER backbone.
    }
    \Description{Pareto frontiers of the additive joint objective
      and training time at top-10 recommendation.}
    \label{fig:pareto_utility_time}
  \end{minipage}
\end{figure*}

\subsection{Analysis of KL-MOPD} \label{sec:methodology:analysis}

\noindentparnolinenospace{Theoretical Guarantees.}
Despite its simplicity, our proposed \mymethod is backed by rigorous theoretical guarantees for joint recommendation objective optimization, as stated in the following:

\begin{tcolorbox}[kessokustyle=bocchipink,breakable=true]
\begin{restatable}{theorem}{RestateTheoremKLMOPDJointObjective} \label{theorem:kl-mopd-joint-objective}
    \upshape
    {(\textit{Joint Objective Improvement of \mymethod}).}
    Let $\pi_{cv, \mathrm{KL}}$ and $\pi_{cv, \mathrm{U}}$ denote the optimal policies of \mymethod in \eqnref{eq:kl-mopd-loss} and vanilla MOPD in \eqnref{eq:vanilla-mopd-loss}, respectively. We can establish a lower bound on the additive joint recommendation objective improvement of \mymethod over vanilla MOPD as:
    \begin{equation} \label{eq:main-bound}
        \mathbb{E}_{\mathcal{X}, \pi_{cv, \mathrm{KL}}}[u_c + u_v] - \mathbb{E}_{\mathcal{X}, \pi_{cv, \mathrm{U}}}[u_c + u_v] \geq \beta S_{\mathrm{KL}} - \frac{D_T L^3}{\beta^2},
    \end{equation}
    where $S_{\mathrm{KL}}$ is a non-negative squared leading term defined as
    \begin{equation} \label{eq:square-gain}
        S_{\mathrm{KL}} \coloneqq \sum_{t=1}^{T} \frac{(d_{c, t} - d_{v, t})^2}{d_{c, t} + d_{v, t}} \geq 0,
    \end{equation}
    while $D_T = \frac{19T}{24} + \frac{1}{2}\sum_{t=1}^{T}\sqrt{t-1}$ is a constant depending only on $T$, and $L = \max\{L_c, L_v\}$. Therefore, when $S_{\mathrm{KL}} > D_T L^3 / \beta^3$, \mymethod is guaranteed to strictly outperform vanilla MOPD, \ie $\mathbb{E}_{\mathcal{X}, \pi_{cv, \mathrm{KL}}}[u_c + u_v] > \mathbb{E}_{\mathcal{X}, \pi_{cv, \mathrm{U}}}[u_c + u_v]$.
\end{restatable}
\end{tcolorbox}

\noindent
The proof is provided in \cref{app:proof-theorem:kl-mopd-joint-objective}. \cref{theorem:kl-mopd-joint-objective} establishes a rigorous lower bound on the performance gain of \mymethod over vanilla MOPD in optimizing the additive joint objective \eqnref{eq:additive-joint-objective}. In practice, the condition $S_{\mathrm{KL}} > D_T L^3 / \beta^3$ is readily satisfied,\footnote{For instance, on TencentGR~\citep{pan2026tencent} ($L = 1, T = 5$), we have $D_T \approx 7.03$, $S_{\mathrm{KL}} \approx 0.05$, and $\beta \approx 6.48$, which clearly satisfies $S_{\mathrm{KL}} > D_T L^3 / \beta^3$. Here, $\beta$ is estimated by Monte Carlo method, \ie sampling item pairs $(y_1, y_2)$ and utilizing the optimality condition $u_j(x, y_1) - u_j(x, y_2) = \beta \log \frac{\pi_j(y_1 | x)}{\pi_{\theta_0}(y_1 | x)} - \beta \log \frac{\pi_j(y_2 | x)}{\pi_{\theta_0}(y_2 | x)}$ from \eqnref{eq:optimal-teacher-policy}.} guaranteeing a strictly positive joint objective gain and thus theoretically validating the effectiveness of \mymethod.

Furthermore, looking closer at the bound \eqnref{eq:main-bound}, the non-negative leading term $S_{\mathrm{KL}}$ in \eqnref{eq:square-gain} is positively correlated with $(d_{c, t} - d_{v, t})^2$, which directly measures the discrepancy between the distillation difficulties of the two teachers at each token $t$. This theoretically echoes our motivation in \cref{sec:preliminaries:challenges}: there exists a token-level imbalance in distillation difficulties between the click and value teachers, and this is precisely what enables \mymethod to achieve superior multi-objective distillation over vanilla MOPD in GR.

\noindentparnoline{Practical Advantages.}
Beyond theoretical guarantees, \mymethod enjoys two appealing advantages for practical deployment:
(i) \emph{Plug-and-Play Design.} Rather than introducing auxiliary networks or complex gradient manipulations, \mymethod simply replaces the uniform weights in \eqnref{eq:vanilla-mopd-loss} with discrepancy-aware token-wise weights $\alpha_{j, t}$ in \eqnref{eq:kl-mopd-loss}, allowing it to be seamlessly integrated into any existing MOPD pipeline and GR backbone.
(ii) \emph{Computational Efficiency.} Since the token-wise discrepancies $d_{j, t}$ and weights $\alpha_{j, t}$ depend solely on the fixed initial student $\pi_{\theta_0}$ and teachers $\pi_j$, they only need to be pre-computed once prior to training. During on-policy distillation, \mymethod requires no additional model forward passes or memory overhead compared to vanilla MOPD, incurring negligible computational cost in practice.


\section{Experiments} \label{sec:experiments}

\subsection{Experimental Setup} \label{sec:experiments:settings}

\noindentparnolinenospace{Datasets.}
We adopt three public datasets, namely TencentGR~\citep{pan2026tencent}, YooChoose~\citep{ben2015recsys}, and Taobao~\citep{zhu2018learning}, alongside a billion-scale industrial advertising dataset (\ie Industry) (\cf \cref{tab:dataset-statistics} for statistics). Each dataset provides two types of feedback for the click ($u_c$) and value ($u_v$) objectives: CTR and CVR on TencentGR, CTR and GMV on YooChoose, CTR and CVR on Taobao, and CTR and eCPM~\citep{xue2026generative} on Industry. For each public dataset, interactions are sorted chronologically and split into 70\% for pre-training the initial student $\pi_{\theta_0}$ on all feedback, 20\% for fine-tuning the specialized teachers ($\pi_c, \pi_v$ initialized from $\pi_{\theta_0}$) and distilling the student $\pi_\theta$, and 10\% for evaluation. For Industry, the fine-tuning and evaluation sets are the last two days of the dataset, each containing more than 11M requests.

\noindentparnoline{Backbones.}
We adopt two representative GR backbones, including TIGER~\citep{rajput2023recommender} and REG4Rec~\citep{xing2026reg4rec}. For each backbone, we closely follow the original implementation and further adapt it to our datasets. Specifically, to construct semantic tokens for each item, following common practice~\citep{si2024generative}, we first learn continuous item embeddings using a 2-layer MLP-based two-tower collaborative filtering model, then quantize them into discrete semantic tokens via RQ-kmeans~\citep{luo2025qarm}, with codebook configurations ($T \times \text{codebook size}$) of $5 \times 64$ on TencentGR, $5 \times 128$ on YooChoose, $8 \times 64$ on Taobao, and $8 \times 310$ on Industry. For the generative recommender, we use a standard encoder-decoder Transformer with 2 encoder layers, 2 decoder layers, and hidden dimension 64 for public datasets, and 5 decoder layers and hidden dimension 1024 for Industry.

\noindentparnoline{Compared Methods.}
We compare \mymethod against the initial student, both single-objective teachers (\ie Backbone-C and Backbone-V), and eight state-of-the-art baselines (mostly from 2026) across four diverse categories:
(i) \emph{traditional recommendation distillation}: RCE-KD~\citep{zhu2026rejuvenating} and MTS$^2$4MM~\citep{ijcai2026p348};
(ii) \emph{GR-specific distillation}: SmartGR~\citep{zhang2026smartgr};
(iii) \emph{value-aware GR}: GR4AD~\citep{xue2026generative} and TSGR~\citep{zhan2026tsgr}; and
(iv) \emph{multi-objective on-policy distillation}: vanilla MOPD~\citep{ma2026mopd}, SMOPD~\citep{wang2026smopd}, and TrustMOPD~\citep{sun2026distill}.
For all the baseline methods, we closely follow the original implementations and further tune hyperparameters for optimal performance on our datasets.

\noindentparnoline{Training Details.}
Following standard practice in prior works~\citep{rajput2023recommender,yang2026bear}, we apply a sliding window over user interaction sequences to construct training and evaluation instances, with a maximum history length of 100 on public datasets and 128 on Industry. We utilize the AdamW optimizer and train all models by update steps until convergence across all stages. Specifically, the initial student $\pi_{\theta_0}$ is pre-trained with a batch size of 1,024 and a learning rate of $10^{-3}$. During both teacher fine-tuning and student distillation, we freeze the encoder and only update the decoder with a batch size of 256, searching the learning rate in $\{10^{-4}, 10^{-3}\}$ and weight decay in $\{10^{-2}, 10^{-3}\}$. The teachers are fine-tuned via SFT for binary objectives (\eg CTR and CVR) and DPO~\citep{rafailov2023direct} for continuous-valued objectives (\eg GMV and eCPM), while all OPD methods sample 4 trajectories per request during distillation.

\noindentparnoline{Evaluation Metrics.}
We evaluate both individual objectives $u_j@K$ ($j \in \{c, v\}$, normalized to $[0, 1]$) and the additive joint objective $u_{cv}@K = (u_c@K + u_v@K) / 2$. Specifically, for each objective $j \in \{c, v\}$, $u_j@K$ is defined as the Recall@$K$ metric weighted by the corresponding objective $u_j(x, y)$. During inference, we adopt constrained beam search~\citep{rajput2023recommender,yang2026bear} and report $K \in \{10, 20\}$ (beam size 50) on public datasets and $K \in \{100, 200\}$ (beam size 200) on Industry.


\begin{table}[t]
  \centering
  \caption{
    Ablation study on teacher weighting strategies with TIGER. Specifically, "+Item-wise KL-MOPD" uses token-shared weights from student-teacher item-wise KL; "+MSE-MOPD" computes token-wise weights via mean squared error between student and teacher token probabilities; and "+Top-$K$ KL-MOPD" restricts KL calculation to the renormalized top-$K$ tokens of the initial student, where $K=16$.
  }
  \label{tab:ablation}
  \footnotesize
  \setlength{\tabcolsep}{2pt}
  \renewcommand{\arraystretch}{1.02}
  \begin{tabularx}{\columnwidth}{l|YY>{\columncolor[rgb]{0.93,0.96,1.00}}Y|YY>{\columncolor[rgb]{0.93,0.96,1.00}}Y}
    \toprule
    \multicolumn{1}{c|}{\multirow{2}[4]{*}{\textbf{Methods}}} & \multicolumn{3}{c|}{\textbf{TencentGR}} & \multicolumn{3}{c}{\textbf{YooChoose}} \\
    \cmidrule(lr){2-4} \cmidrule(lr){5-7}
    & \textbf{\boldmath$u_c$@10} & \textbf{\boldmath$u_v$@10} & \textbf{\boldmath$u_{cv}$@10} & \textbf{\boldmath$u_c$@10} & \textbf{\boldmath$u_v$@10} & \textbf{\boldmath$u_{cv}$@10} \\
    \midrule
    +MOPD & 0.0261 & 0.0164 & 0.0213 & 0.4186 & 0.5265 & 0.4726 \\
    \cmidrule(lr){1-1} \cmidrule(lr){2-4} \cmidrule(lr){5-7}
    +Item-wise KL-MOPD & 0.0255 & 0.0164 & 0.0209 & 0.4204 & 0.5424 & 0.4814 \\
    +MSE-MOPD & 0.0262 & 0.0164 & 0.0213 & 0.4178 & 0.5525 & 0.4851 \\
    +Top-$K$ KL-MOPD & 0.0244 & 0.0170 & 0.0207 & 0.4180 & 0.5516 & 0.4848 \\
    \cmidrule(lr){1-1} \cmidrule(lr){2-4} \cmidrule(lr){5-7}
    \textbf{+\mymethod} & \textbf{0.0294} & \textbf{0.0186} & \textbf{0.0240} & \textbf{0.4327} & \textbf{0.5741} & \textbf{0.5034} \\
    \bottomrule
  \end{tabularx}
\end{table}


\subsection{Experimental Results} \label{sec:experiments:results}

\noindentparnolinenospace{Overall Performance.}
\cref{tab:main-results} shows the overall performance across three public datasets and two GR backbones. \mymethod consistently outperforms all baselines on the individual objectives $u_c$ and $u_v$ by an average of \emph{+4.91\%}. More importantly, on the joint objective $u_{cv}$, \mymethod achieves the best performance in all settings, surpassing both the strongest baseline and single-objective teachers by \emph{+4.12\%} on average, validating our theoretical analysis in \cref{theorem:kl-mopd-joint-objective}. Furthermore, as shown in \cref{fig:pareto_click_value}, \mymethod establishes a strictly \emph{superior Pareto frontier} between $u_c$ and $u_v$ compared to representative baselines, confirming its effectiveness in balancing dual-teacher distillation process across decoding steps.

\noindentparnoline{Ablation Study.}
To further validate the effectiveness of our proposed token-wise weighting scheme, we conduct an ablation study in \cref{tab:ablation}, including three alternative weighting variants:
(i) \emph{+Item-wise KL-MOPD}, which uses token-shared weights computed from the item-wise KL divergence;
(ii) \emph{+MSE-MOPD}, which computes token-wise weights via mean squared discrepancy between student and teacher token probabilities; and
(iii) \emph{+Top-$K$ KL-MOPD}, which restricts the KL calculation to the renormalized top-16 tokens of $\pi_{\theta_0}$.
As shown, \mymethod consistently outperforms all three variants, demonstrating the rationale of our proposed weighting scheme.

\noindentparnoline{Computational Efficiency.}
\cref{fig:pareto_utility_time} plots the joint performance $u_{cv}$ against training time per step. Since \mymethod pre-computes $\alpha_{j, t}$ only once prior to training, it incurs virtually zero extra overhead over MOPD. Compared to baselines requiring dynamic weight updates (\eg TrustMOPD) or extra forward passes (\eg RCE-KD), \mymethod achieves better accuracy-efficiency trade-off.

\noindentparnoline{Industrial Offline and Online Evaluation.}
Finally, we validate \mymethod in real-world e-commerce advertising scenarios. In offline evaluations on the billion-scale Industry dataset (\cf \cref{tab:industrial-offline}), \mymethod consistently outperforms vanilla MOPD and both teachers in single- and joint-objective metrics (avg. \emph{+3.19\%} on $u_{cv}$). Furthermore, we deploy \mymethod in a 5-day online A/B test on our production advertising platform, allocating 10\% live traffic to the control group and 10\% to our treatment group. As reported in \cref{tab:industrial-online}, \mymethod achieves statistically significant uplifts ($p < 0.05$) of \emph{+1.63\%} in advertising revenue (REV) and \emph{+2.20\%} in GMV, confirming its practical value in large-scale industrial applications.




\begin{table}[t]
  \centering
  \caption{
    Offline results on the large-scale industrial dataset with TIGER. Given the large number of ads, we evaluate performance on top-100 and top-200 recommendations. Due to limited resources, we only compare \mymethod with MOPD. Formatting and improvement definitions follow \cref{tab:main-results}.
  }
  \label{tab:industrial-offline}
  \setlength{\tabcolsep}{3pt}
  \resizebox{\columnwidth}{!}{%
    \begin{tabular}{l|cc>{\columncolor[rgb]{0.93,0.96,1.00}}ccc>{\columncolor[rgb]{0.93,0.96,1.00}}c}
    \toprule
    \multicolumn{1}{c|}{\multirow{2}[4]{*}{\textbf{Methods}}} & \multicolumn{6}{c}{\textbf{Industry}} \\
    \cmidrule(lr){2-7}
    & \textbf{\boldmath$u_{c}$@100} & \textbf{\boldmath$u_{v}$@100} & \textbf{\boldmath$u_{cv}$@100} & \textbf{\boldmath$u_{c}$@200} & \textbf{\boldmath$u_{v}$@200} & \textbf{\boldmath$u_{cv}$@200} \\
    \midrule
    \textbf{TIGER} & 0.3934 & 0.1695 & 0.2815 & 0.4698 & 0.2187 & 0.3443 \\
    TIGER-C & 0.4090 & 0.1757 & 0.2924 & \uline{0.4772} & 0.2238 & 0.3505 \\
    TIGER-V & 0.4088 & 0.1968 & 0.3028 & 0.4494 & 0.2308 & 0.3401 \\
    \cmidrule(lr){1-1} \cmidrule(lr){2-7}
    +MOPD & \uline{0.4309} & \uline{0.1969} & \uline{0.3139} & 0.4729 & \uline{0.2366} & \uline{0.3548} \\
    \cmidrule(lr){1-1} \cmidrule(lr){2-7}
    \textbf{+\mymethod} & \textbf{0.4417} & \textbf{0.2054} & \textbf{0.3236}$^{*}$ & \textbf{0.4860} & \textbf{0.2469} & \textbf{0.3665}$^{*}$ \\
    \textcolor{darkred}{\textbf{Imp.\%}} & \textcolor{darkred}{\textbf{+2.51\%}} & \textcolor{darkred}{\textbf{+4.32\%}} & \textcolor{darkred}{\textbf{+3.07\%}$^{*}$} & \textcolor{darkred}{\textbf{+2.77\%}} & \textcolor{darkred}{\textbf{+4.35\%}} & \textcolor{darkred}{\textbf{+3.30\%}$^{*}$} \\
    \bottomrule
    \end{tabular}
  }
\end{table}

\begin{table}[t]
  \centering
  \caption{
    Online A/B test results of \mymethod on the real-world industrial e-commerce advertising platform, in advertising revenue (REV) and gross merchandise volume (GMV).
  }
  \label{tab:industrial-online}
  \resizebox{\columnwidth}{!}{%
    \begin{tabular}{l|c|c}
    \toprule
    \textbf{Metrics} & \makecell{\textbf{Advertising Revenue}\\\textbf{(REV)}} & \makecell{\textbf{Gross Merchandise Volume}\\\textbf{(GMV)}} \\
    \midrule
    \textcolor{darkred}{\textbf{Imp.\%}} & \textcolor{darkred}{\textbf{+1.63\%}} & \textcolor{darkred}{\textbf{+2.20\%}} \\
    \bottomrule
    \end{tabular}
  }
\end{table}


\section{Related Work} \label{sec:related_work}

\noindentparnolinenospace{Generative Recommendation.}
Generative recommendation (GR) reformulates item retrieval as autoregressive generation over discrete semantic tokens, achieving strong scalability and generalization~\citep{rajput2023recommender,wang2024learnable,zheng2024adapting,si2024generative,deng2025onerec,xing2026reg4rec}. However, recent studies reveal that GR remains challenging to optimize: unlike natural language generation, successfully retrieving a positive item in GR strictly requires making accurate decoding decisions at every token~\citep{yang2026bear,yu2026apao,zhang2026smartgr}. In this work, we explicitly account for this specific requirement via balanced token-wise optimization across all decoding steps.

\noindentparnoline{Multi-Objective Optimization.}
Simultaneously learning multiple heterogeneous objectives is a ubiquitous requirement in machine learning tasks~\citep{huang2024data,weiss2025clicks}. Existing approaches broadly fall into multi-task architecture design~\citep{ma2018entire,ma2018modeling,tang2020progressive}, loss scalarization and Pareto-efficient gradient manipulation~\citep{sener2018multi,lin2019pareto,yu2020gradient,liu2021conflict,xie2021personalized}, and multi-objective preference alignment or score fusion~\citep{pei2019value,gu2020deep,zeng2021enhancing,zhou2024beyond,cao2025xmtf,he2025end,meng2025generative}. Recently, several studies have also incorporated value-aware optimization into GR, such as GR4AD~\citep{xue2026generative}, TSGR~\citep{zhan2026tsgr}, UniVA~\citep{zhang2026unified}, and RGD~\citep{yang2026reward}. Nevertheless, these methods primarily coordinate objectives at the coarse item-wise level, overlooking the multi-token generation nature of GR. In contrast, our \mymethod explicitly coordinates multiple objectives at the fine-grained token level.

\noindentparnoline{Recommendation Distillation.}
Knowledge distillation has been widely used in recommendation to transfer teacher capabilities into student via item rankings~\citep{tang2018ranking,liu2020developing,kang2020rrd,yang2022cross,kang2023distillation,zhu2026rejuvenating}, hidden representations~\citep{kang2021topology,zhu2025exploring,ijcai2026p348}, and prediction scores~\citep{xv2022lightweight,tang2024multi,yi2024dmmp}, which mostly operate at item-level. Recently, studies have begun adapting distillation to GR~\citep{zhang2026smartgr}, notably by introducing OPD~\citep{team2026onereason} for its advantages of on-policy learning and token-wise supervision. Our \mymethod inherits these merits while further resolving the token-level imbalance and cross-token inconsistency across multiple teachers.

\noindentparnoline{On-Policy Distillation.}
Recently, OPD has emerged as a mainstream post-training paradigm~\citep{agarwal2024policy,gu2024minillm}. Beyond refinements to single-teacher OPD~\citep{li2026filter,wang2026dr,wang2026metaopd}, MOPD~\citep{ma2026mopd} has attracted widespread interest for its success in merging multiple specialized experts. To balance multiple objectives, recent studies have explored incorporating reinforcement learning (\eg SMOPD~\citep{wang2026smopd}) or dynamic weighting (\eg Open-MOPD~\citep{gao2026open} and TrustMOPD~\citep{sun2026distill}). However, these methods are mostly intuitive and not tailored to recommendation. In contrast, our \mymethod captures the token-level imbalance and cross-token inconsistency in GR and allocates token-wise weights via initial student-teacher discrepancy, optimizing the joint objective with rigorous theoretical guarantees over vanilla MOPD.


\section{Conclusion and Future Directions} \label{sec:conclusion_and_future_directions}

This work explores multi-objective on-policy distillation (MOPD) for generative recommendation to coordinate competing objectives through fine-grained token-wise supervision. Through empirical analysis, we uncover that vanilla MOPD suffers from token-level imbalance and cross-token inconsistency, which hinder positive item retrieval for both objectives. To address this challenge, we propose \mymethod, a plug-and-play and computationally efficient method that allocates token-wise teacher weights proportional to the initial student-teacher KL divergences. Notably, \mymethod is backed by rigorous theoretical guarantees, strictly outperforming vanilla MOPD in the additive joint objective under mild conditions, and its effectiveness is extensively validated across public benchmarks, a billion-scale industrial dataset, and production online A/B tests. Looking forward, promising directions include: (i) extending our theoretical framework to scenarios with three or more teachers, and (ii) generalizing \mymethod to general language tasks.






\bibliographystyle{ACM-Reference-Format}
\bibliography{references}




\newpage    
\appendix


\section{Proof of \cref{theorem:kl-mopd-joint-objective}} \label{app:proof-theorem:kl-mopd-joint-objective}

In this section, we provide the complete proof of \cref{theorem:kl-mopd-joint-objective}. We first introduce some useful notations and preliminary results in \cref{app:proof-theorem:kl-mopd-joint-objective:preliminaries}, and then present the main proof in \cref{app:proof-theorem:kl-mopd-joint-objective:main-proof}.

\subsection{Preliminary Results} \label{app:proof-theorem:kl-mopd-joint-objective:preliminaries}

\noindentparnolinenospace{Notation Simplification.}
For notation simplicity, for any policy $\pi$ and function $f$, we denote the expectation over user requests and policy rollouts as $\mathbb{E}_{\mathcal{X}, \pi}[f] \coloneqq \mathbb{E}_{x \sim \mathcal{X}, \, y \sim \pi(\cdot | x)}[f(x, y)]$. We further denote $L \coloneqq \max\{L_c, L_v\}$, and $\varnothing$ as the empty prefix.

\noindentparnoline{Optimal Students with MOPD.}
From \eqnref{eq:optimal-teacher-policy}, the optimal teacher policy for objective $j \in \{c, v\}$ takes the explicit form $\pi_j(y | x) = \frac{1}{Z_j(x)} \pi_{\theta_0}(y | x) \exp\left( u_j(x, y) / \beta \right)$ with the partition function $Z_j(x) = \sum_{y'} \pi_{\theta_0}(y' | x) \exp\left( u_j(x, y') / \beta \right)$. Then, given any token-wise weights $\alpha_{c, t}, \alpha_{v, t} \in [0, 1]$ with $\alpha_{c, t} + \alpha_{v, t} = 1$, minimizing the weighted token-wise KL divergence $\sum_{j \in \{c, v\}} \alpha_{j, t} \mathbb{D}_{\mathrm{KL}} \left( \pi_\theta(\cdot | x, y_{<t}) \,\|\, \pi_j(\cdot | x, y_{<t}) \right)$ at each prefix $(x, y_{<t})$ yields the optimal token-wise target policy $\pi_{cv}$, which is the normalized geometric mixture of the two teachers:
\begin{equation} \label{eq:conditional-target}
    \pi_{cv}(y_t | x, y_{<t}) = \frac{\left( \pi_c(y_t | x, y_{<t}) \right)^{\alpha_{c, t}} \left( \pi_v(y_t | x, y_{<t}) \right)^{\alpha_{v, t}}}{Z_{cv}(x, y_{<t})},
\end{equation}
where $Z_{cv}(x, y_{<t})$ is the partition function. The corresponding item-wise policy is defined as $\pi_{cv}(y | x) = \prod_{t=1}^{T} \pi_{cv}(y_t | x, y_{<t})$. Specifically, $\pi_{cv, \mathrm{KL}}$ and $\pi_{cv, \mathrm{U}}$ denote the optimal policies under the \mymethod's weights $\alpha_{j, t} = \frac{d_{j, t}}{d_{c, t} + d_{v, t}}$ in \eqnref{eq:weight} and the vanilla MOPD's uniform weights $\alpha_{j, t} = 1/2$ in \eqnref{eq:vanilla-mopd-loss}, respectively.

\noindentparnoline{Conditional Expected Objective, Increment, and Variance.}
For any request $x$ and prefix $y_{<t}$, we define the \emph{conditional expected objective} as $U_j(x, y_{<t}) \coloneqq \mathbb{E}_{y' \sim \pi_{\theta_0}(\cdot | x)} \left[ u_j(x, y') \mid y'_{<t} = y_{<t} \right]$ under the initial student $\pi_{\theta_0}$, where $U_j(x, \varnothing) = \mathbb{E}_{y \sim \pi_{\theta_0}(\cdot | x)}[u_j(x, y)]$ for the empty prefix $\varnothing$ and $U_j(x, y_{\leq T}) = u_j(x, y)$. We further define the token-wise \emph{conditional expected objective increment} as the change in conditional expected objective $U_j(x, \cdot)$ after observing token $y_t$:
\begin{equation} \label{eq:objective-increment}
    \Delta U_j(x, y_{\leq t}) \coloneqq U_j(x, y_{\leq t}) - U_j(x, y_{<t}).
\end{equation}
By the law of total expectation, the expected increment is zero, \ie $\mathbb{E}_{y_t \sim \pi_{\theta_0}(\cdot | x, y_{<t})} \left[ \Delta U_j(x, y_{\leq t}) \right] = 0$. Consequently, the expected conditional variance of this increment at step $t$ is given by:
\begin{equation} \label{eq:increment-variance}
\begin{aligned}
    v_{j, t} & \coloneqq \mathbb{E}_{\substack{x \sim \mathcal{X} \\ y_{<t} \sim \pi_{\theta_0}(\cdot | x)}} \left[ \mathbb{V}_{y_t \sim \pi_{\theta_0}(\cdot | x, y_{<t})} \left[ \Delta U_j(x, y_{\leq t}) \right] \right] \\
    & = \mathbb{E}_{\substack{x \sim \mathcal{X} \\ y \sim \pi_{\theta_0}(\cdot | x)}} \left[ \left( \Delta U_j(x, y_{\leq t}) \right)^2 \right].
\end{aligned}
\end{equation}

\noindentparnoline{Variance and Covariance Bounds.}
For any real-valued function $f(z)$ on a finite set, we define its oscillation as $\operatorname{osc}(f) \coloneqq \max_z f(z) - \min_z f(z)$. Under any distribution $p$, it satisfies the standard bounds:
\begin{equation} \label{eq:osc-bounds}
    \mathbb{V}_p[f] \leq \frac{\operatorname{osc}(f)^2}{4}, \quad \left| \operatorname{Cov}_p(f, g) \right| \leq \frac{\operatorname{osc}(f)\operatorname{osc}(g)}{4}.
\end{equation}
Based on these bounds, we can establish the following lemma:

\begin{tcolorbox}[kessokustyle=nijikayellow]
\begin{restatable}{lemma}{RestateLemmaExpReweighting} \label{lemma:exponential-reweighting}
    \upshape
    {(\textit{Exponential Reweighting Bounds}).}
    For any distribution $p(z)$ on a finite set $\mathcal{Z}$ and real-valued functions $f(z), g(z)$, define the exponentially reweighted distribution $q(z) = \frac{p(z) e^{g(z)}}{\mathbb{E}_p[e^g]}$. Then, we have the following bounds:
    \begin{equation} \label{eq:exp-reweighting-bounds-1}
        \left| \mathbb{E}_q[f] - \mathbb{E}_p[f] - \operatorname{Cov}_p(f, g) \right| \leq \frac{\operatorname{osc}(f)\operatorname{osc}(g)^2}{8},
    \end{equation}
    and
    \begin{equation} \label{eq:exp-reweighting-bounds-2}
        \mathbb{D}_{\mathrm{KL}}(p \,\|\, q) \leq \frac{\operatorname{osc}(g)^2}{8}.
    \end{equation}
\end{restatable}
\end{tcolorbox}

\begin{proof}[Proof of \cref{lemma:exponential-reweighting}]

For \eqnref{eq:exp-reweighting-bounds-1}, we define the interpolated distribution $p_\lambda(z) = \frac{p(z) e^{\lambda g(z)}}{\mathbb{E}_p[e^{\lambda g}]}$ and $F(\lambda) = \mathbb{E}_{p_\lambda}[f]$ for $\lambda \in [0, 1]$. Direct differentiation yields $F(0) = \mathbb{E}_p[f]$, $F(1) = \mathbb{E}_q[f]$, $F'(0) = \operatorname{Cov}_p(f, g)$, and $F''(\lambda) = \mathbb{E}_{p_\lambda} \left[ (f - \mathbb{E}_{p_\lambda}[f])(g - \mathbb{E}_{p_\lambda}[g])^2 \right]$. By \eqnref{eq:osc-bounds}, we further have $|F''(\lambda)| \leq \operatorname{osc}(f) \mathbb{V}_{p_\lambda}[g] \leq \frac{\operatorname{osc}(f)\operatorname{osc}(g)^2}{4}$. Substituting these results into \eqnref{eq:exp-reweighting-bounds-1} and applying Taylor's formula gives:
\begin{equation}
\begin{aligned}
    & \left| \mathbb{E}_q[f] - \mathbb{E}_p[f] - \operatorname{Cov}_p(f, g) \right| = |F(1) - F(0) - F'(0)|
    \\
    &= \left| \int_0^1 (1 - \lambda) F''(\lambda) \mathrm{d}\lambda \right| \leq \frac{\operatorname{osc}(f)\operatorname{osc}(g)^2}{8}.
\end{aligned}
\end{equation}

For \eqnref{eq:exp-reweighting-bounds-2}, we further define $H(\lambda) = \log \mathbb{E}_p[e^{\lambda g}]$. Since $H(0) = 0$, $H(1) = \log \mathbb{E}_p[e^g]$, $H'(0) = \mathbb{E}_p[g]$, and $0 \leq H''(\lambda) = \mathbb{V}_{p_\lambda}[g] \leq \frac{\operatorname{osc}(g)^2}{4}$, we obtain $\mathbb{D}_{\mathrm{KL}}(p \,\|\, q) = \log \mathbb{E}_p[e^g] - \mathbb{E}_p[g] = H(1) - H(0) - H'(0) = \int_0^1 (1 - \lambda) H''(\lambda) \mathrm{d}\lambda \leq \frac{\operatorname{osc}(g)^2}{8}$. This completes the proof.
\end{proof}

\subsection{Main Proof of \cref{theorem:kl-mopd-joint-objective}} \label{app:proof-theorem:kl-mopd-joint-objective:main-proof}

\noindentparnolinenospace{Proof Overview.}
To prove \cref{theorem:kl-mopd-joint-objective}, we first establish two lemmas. Specifically, \cref{lemma:kl-variance-relation} connects the initial student-teacher token-wise KL divergence $d_{j, t}$ to the objective increment variance $v_{j, t}$. Next, \cref{lemma:joint-objective-expansion} expands the joint objective difference between $\pi_{cv, \mathrm{KL}}$ and $\pi_{cv, \mathrm{U}}$ into the leading term $\frac{1}{\beta} \sum_{t=1}^{T} \left( \alpha_{c, t} - \frac{1}{2} \right) (v_{c, t} - v_{v, t})$ up to a bounded remainder. Finally, substituting \cref{lemma:kl-variance-relation} and our KL-proportional weights $\alpha_{j, t}$ from \eqnref{eq:weight} into \cref{lemma:joint-objective-expansion} directly yields the desired lower bound in \cref{theorem:kl-mopd-joint-objective}.

\begin{tcolorbox}[kessokustyle=nijikayellow]
\begin{restatable}{lemma}{RestateLemmaKLVariance} \label{lemma:kl-variance-relation}
    \upshape
    {(\textit{Approximation of Initial Student-Teacher KL by Objective Increment Variance}).}
    For each objective $j \in \{c, v\}$ and decoding step $t = 1, 2, \cdots, T$, the initial student-teacher token-wise KL divergence $d_{j, t}$ in \eqnref{eq:initial-kl} can be approximated by $d_{j, t} \approx \frac{v_{j, t}}{2\beta^2}$, with the approximation error bounded by:
    \begin{equation} \label{eq:kl-variance-bound}
        \left| d_{j, t} - \frac{v_{j, t}}{2\beta^2} \right| \leq \frac{L_j^3}{12\beta^3}.
    \end{equation}
\end{restatable}
\end{tcolorbox}

\begin{proof}[Proof of \cref{lemma:kl-variance-relation}]
For any request $x$, prefix $y_{<t}$, and $\lambda \in [0, 1]$, define the conditional log-moment generating function as:
\begin{equation} \label{eq:log-mgf}
    \phi_j(\lambda; x, y_{<t}) \coloneqq \log \mathbb{E}_{y' \sim \pi_{\theta_0}(\cdot | x)} \left[ \exp\left( \frac{\lambda u_j(x, y')}{\beta} \right) \;\middle|\; y'_{<t} = y_{<t} \right].
\end{equation}
Since $\pi_j(y | x) = \frac{1}{Z_j(x)} \pi_{\theta_0}(y | x) \exp(u_j(x, y) / \beta)$, taking expectations yields $\pi_j(y_{<t} | x) = \frac{\pi_{\theta_0}(y_{<t} | x)}{Z_j(x)} \exp(\phi_j(1; x, y_{<t}))$, and taking the ratio between $y_{\leq t}$ and $y_{<t}$ gives the token-wise log-probability ratio:
\begin{equation} \label{eq:teacher-token-ratio}
    \log \frac{\pi_j(y_t | x, y_{<t})}{\pi_{\theta_0}(y_t | x, y_{<t})} = \phi_j(1; x, y_{\leq t}) - \phi_j(1; x, y_{<t}).
\end{equation}
Therefore, the token-wise KL divergence can be rewritten as:
\begin{equation} \label{eq:local-kl-mgf}
\begin{aligned}
    &\mathbb{D}_{\mathrm{KL}} \left( \pi_{\theta_0}(\cdot | x, y_{<t}) \,\|\, \pi_j(\cdot | x, y_{<t}) \right) \\
    &\qquad = \phi_j(1; x, y_{<t}) - \mathbb{E}_{y_t \sim \pi_{\theta_0}(\cdot | x, y_{<t})} \left[ \phi_j(1; x, y_{\leq t}) \right].
\end{aligned}
\end{equation}
To estimate the token-wise KL divergence, we expand $\phi_j(\lambda; x, y_{<t})$ at $\lambda = 0$. Specifically, $\phi_j(0; x, y_{<t}) = 0$, $\phi_j'(0; x, y_{<t}) = \frac{1}{\beta} U_j(x, y_{<t})$, and $\phi_j''(0; x, y_{<t}) = \frac{1}{\beta^2} \mathbb{V}_{y' \sim \pi_{\theta_0}(\cdot | x)} \left[ u_j(x, y') \mid y'_{<t} = y_{<t} \right]$. Moreover, since $\phi_j'''(\lambda; x, y_{<t})$ is the third central moment of $u_j / \beta \in [0, L_j / \beta]$ under the tilted conditional distribution, it is bounded by $|\phi_j'''(\lambda; x, y_{<t})| \leq \frac{L_j}{\beta} \cdot \frac{L_j^2}{4\beta^2} = \frac{L_j^3}{4\beta^3}$ via \eqnref{eq:osc-bounds}. Then, applying Taylor's theorem at $\lambda = 1$ yields:
\begin{equation} \label{eq:mgf-taylor}
\begin{aligned}
    & \phi_j(1; x, y_{<t}) = \frac{1}{\beta} U_j(x, y_{<t}) \\
    & + \frac{1}{2\beta^2} \mathbb{V}_{y' \sim \pi_{\theta_0}(\cdot | x)} \left[ u_j(x, y') \mid y'_{<t} = y_{<t} \right] + R_j(x, y_{<t}),
\end{aligned}
\end{equation}
where the remainder $|R_j(x, y_{<t})| \leq \frac{1}{6} \sup_{\lambda \in [0, 1]} |\phi_j'''(\lambda; x, y_{<t})| \leq \frac{L_j^3}{24\beta^3}$. Substituting \eqnref{eq:mgf-taylor} into \eqnref{eq:local-kl-mgf}, the first-order terms cancel by $\mathbb{E}_{y_t \sim \pi_{\theta_0}(\cdot | x, y_{<t})}[U_j(x, y_{\leq t})] = U_j(x, y_{<t})$, then applying the law of total variance to the second-order terms further gives:
\begin{equation} \label{eq:kl-variance-relation}
\begin{aligned}
    & \mathbb{D}_{\mathrm{KL}} \left( \pi_{\theta_0}(\cdot | x, y_{<t}) \,\|\, \pi_j(\cdot | x, y_{<t}) \right) \\
    & = \frac{1}{2\beta^2} \mathbb{V}_{y_t \sim \pi_{\theta_0}(\cdot | x, y_{<t})} \left[ U_j(x, y_{\leq t}) \right] + \Delta R_j(x, y_{<t}) \\
    & = \frac{1}{2\beta^2} \mathbb{V}_{y_t \sim \pi_{\theta_0}(\cdot | x, y_{<t})} \left[ \Delta U_j(x, y_{\leq t}) \right] + \Delta R_j(x, y_{<t}),
\end{aligned}
\end{equation}
where $\Delta R_j(x, y_{<t}) \coloneqq R_j(x, y_{<t}) - \mathbb{E}_{y_t \sim \pi_{\theta_0}(\cdot | x, y_{<t})} \left[ R_j(x, y_{\leq t}) \right]$ satisfies $|\Delta R_j(x, y_{<t})| \leq \frac{L_j^3}{12\beta^3}$, and the last equality holds because $U_j(x, y_{<t})$ is constant \wrt $y_t$. Taking expectations of both sides in \eqnref{eq:kl-variance-relation} over $x \sim \mathcal{X}, \, y \sim \pi_{\theta_0}(\cdot | x)$, we obtain the desired \eqnref{eq:kl-variance-bound}.
\end{proof}

\begin{tcolorbox}[kessokustyle=nijikayellow]
\begin{restatable}{lemma}{RestateLemmaJointExpansion} \label{lemma:joint-objective-expansion}
    \upshape
    {(\textit{Approximation of Joint Objective Difference}).}
    Denote the joint objective difference between the policies $\pi_{cv}$ and $\pi_{cv, \mathrm{U}}$ as $\Delta_{\pi_{cv}, \pi_{cv, \mathrm{U}}} [u_c + u_v] \coloneqq \mathbb{E}_{\mathcal{X}, \pi_{cv}}[u_c + u_v] - \mathbb{E}_{\mathcal{X}, \pi_{cv, \mathrm{U}}}[u_c + u_v]$. Then, we have the approximation $\Delta_{\pi_{cv}, \pi_{cv, \mathrm{U}}} [u_c + u_v] \approx \frac{1}{\beta} \sum_{t=1}^{T} \left( \alpha_{c, t} - \frac{1}{2} \right) (v_{c, t} - v_{v, t})$ for any normalized weights $\alpha_{j, t}$, with the error bounded by:
    \begin{equation} \label{eq:joint-expansion-bound}
    \begin{aligned}
        & \left| \Delta_{\pi_{cv}, \pi_{cv, \mathrm{U}}} [u_c + u_v] - \frac{1}{\beta} \sum_{t=1}^{T} \left( \alpha_{c, t} - \frac{1}{2} \right) (v_{c, t} - v_{v, t}) \right| \\
        & \qquad \leq \frac{L^3}{\beta^2} \left( \frac{5T}{8} + \frac{1}{2} \sum_{t=1}^{T} \sqrt{t-1} \right).
    \end{aligned}
    \end{equation}
\end{restatable}
\end{tcolorbox}

\begin{proof}[Proof of \cref{lemma:joint-objective-expansion}]
\noindentparnolinenospace{Rewriting the Joint Objective Difference.}
We first express $u_j$ using the objective increments $\Delta U_j$ in \eqnref{eq:objective-increment}. Given $\sum_{t=1}^{T} \Delta U_j(x, y_{\leq t}) = U_j(x, y_{\leq T}) - U_j(x, \varnothing) = u_j(x, y) - \mathbb{E}_{y' \sim \pi_{\theta_0}(\cdot | x)}[u_j(x, y')]$, taking expectations over $\mathcal{X}, \pi_{cv}$ yields:
\begin{equation} \label{eq:telescoping-joint-diff}
\begin{aligned}
    & \mathbb{E}_{\mathcal{X}, \pi_{cv}}[u_c + u_v] - \mathbb{E}_{\mathcal{X}, \pi_{\theta_0}}[u_c + u_v] \\
    & = \sum_{t=1}^{T} \mathbb{E}_{\substack{x \sim \mathcal{X} \\ y_{<t} \sim \pi_{cv}(\cdot | x)}} \left[ \mathbb{E}_{y_t \sim \pi_{cv}(\cdot | x, y_{<t})} \left[ \Delta U_c(x, y_{\leq t}) + \Delta U_v(x, y_{\leq t}) \right] \right].
\end{aligned}
\end{equation}

\noindentparnoline{Rewriting the Optimal Policy with Weighted MOPD.}
Next, we can express $\pi_{cv}(y_t | x, y_{<t})$ in an exponential reweighting form of $\pi_{\theta_0}(y_t | x, y_{<t})$ (\cf \cref{lemma:exponential-reweighting}). Recall in \eqnref{eq:teacher-token-ratio} that $\pi_j(y_t | x, y_{<t}) = \pi_{\theta_0}(y_t | x, y_{<t}) \exp(\phi_j(1; x, y_{\leq t}) - \phi_j(1; x, y_{<t}))$, substituting the Taylor expansion $\phi_j(1; x, y_{\leq t}) = \frac{1}{\beta} U_j(x, y_{\leq t}) + \int_0^1 (1 - \lambda) \phi_j''(\lambda; x, y_{\leq t}) \mathrm{d}\lambda$ into \eqnref{eq:conditional-target} and absorbing $y_t$-independent terms into the partition function, we can express the optimal policy $\pi_{cv}(y_t | x, y_{<t})$:
\begin{equation} \label{eq:target-exp-form}
    \pi_{cv}(y_t | x, y_{<t}) \propto \pi_{\theta_0}(y_t | x, y_{<t}) \exp\left( g(x, y_{\leq t}) \right),
\end{equation}
where $g(x, y_{\leq t}) \coloneqq \frac{1}{\beta} \left( \alpha_{c, t} \Delta U_c(x, y_{\leq t}) + \alpha_{v, t} \Delta U_v(x, y_{\leq t}) \right) + R(x, y_{\leq t})$, with remainder $R(x, y_{\leq t}) = \sum_{j \in \{c, v\}} \alpha_{j, t} \int_0^1 (1 - \lambda) \phi_j''(\lambda; x, y_{\leq t}) \mathrm{d}\lambda$. Since $0 \leq \phi_j''(\lambda; x, y_{\leq t}) \leq \frac{L_j^2}{4\beta^2}$, we have $0 \leq R(x, y_{\leq t}) \leq \frac{L^2}{8\beta^2}$. Furthermore, note that the log-moment generating function $\phi_j \in [0, L_j / \beta]$, the original $\phi$-expression of $g$ reveals that $\operatorname{osc}(g) \leq \frac{L}{\beta}$.

\noindentparnoline{Approximating the Inner Expectation by Covariance.}
Then, we are moving on approximating the inner expectation in \eqnref{eq:telescoping-joint-diff}, \ie $\mathbb{E}_{y_t \sim \pi_{cv}(\cdot | x, y_{<t})} \left[ \Delta U_c(x, y_{\leq t}) + \Delta U_v(x, y_{\leq t}) \right]$, by the covariance term. Specifically, we can fix $(x, y_{<t})$ and consider $y_t$ as the random variable. Then, we apply \cref{lemma:exponential-reweighting} with $p = \pi_{\theta_0}(\cdot | x, y_{<t})$, $q = \pi_{cv}(\cdot | x, y_{<t})$, $f = \Delta U_c(x, y_{\leq t}) + \Delta U_v(x, y_{\leq t})$, and $g = g(x, y_{\leq t})$. Note that $\mathbb{E}_p[f] = 0$, and since $U_j(x, y_{<t})$ is constant \wrt $y_t$, $\operatorname{osc}(\Delta U_j) = \operatorname{osc}(U_j(x, y_{\leq t})) \leq L_j$, which implies $\operatorname{osc}(f) \leq L_c + L_v \leq 2L$. Thus, \cref{lemma:exponential-reweighting} gives $\left| \mathbb{E}_q[f] - \operatorname{Cov}_p(f, g) \right| \leq \frac{2L \cdot (L/\beta)^2}{8} = \frac{L^3}{4\beta^2}$. We further decompose $\operatorname{Cov}_p(f, g) = C(x, y_{<t}) + \operatorname{Cov}_p(f, R)$, where the main covariance term $C(x, y_{<t})$ is defined as:
\begin{equation} \label{eq:local-covariance-def}
    C(x, y_{<t}) \coloneqq \frac{1}{\beta} \operatorname{Cov}_{y_t \sim \pi_{\theta_0}(\cdot | x, y_{<t})} \left( \Delta U_c + \Delta U_v, \, \alpha_{c, t} \Delta U_c + \alpha_{v, t} \Delta U_v \right),
\end{equation}
and recall that $R \in [0, L^2 / (8\beta^2)]$, \eqnref{eq:osc-bounds} bounds the remainder covariance as $|\operatorname{Cov}_p(f, R)| \leq \frac{2L \cdot (L^2 / (8\beta^2))}{4} = \frac{L^3}{16\beta^2}$, therefore:
\begin{equation} \label{eq:local-covariance-approx}
    \left| \mathbb{E}_{y_t \sim \pi_{cv}(\cdot | x, y_{<t})} \left[ \Delta U_c(x, y_{\leq t}) + \Delta U_v(x, y_{\leq t}) \right] - C(x, y_{<t}) \right| \leq \frac{5L^3}{16\beta^2}.
\end{equation}

\noindentparnoline{Controlling Prefix Distribution Shift.}
Note that the inner expectation in \eqnref{eq:telescoping-joint-diff} is taken under $y_{<t} \sim \pi_{cv}(\cdot | x)$, while the covariance term $C(x, y_{<t})$ is defined under $y_{<t} \sim \pi_{\theta_0}(\cdot | x)$. Therefore, we need to estimate the error incurred by this prefix distribution shift. Specifically, by \cref{lemma:exponential-reweighting}, the token-wise KL satisfies $\mathbb{D}_{\mathrm{KL}} \left( \pi_{\theta_0}(\cdot | x, y_{<\tau}) \,\|\, \pi_{cv}(\cdot | x, y_{<\tau}) \right) \leq \frac{\operatorname{osc}(g)^2}{8} \leq \frac{L^2}{8\beta^2}$, and the chain rule yields $\mathbb{D}_{\mathrm{KL}} \left( \pi_{\theta_0}(y_{<t} | x) \,\|\, \pi_{cv}(y_{<t} | x) \right) \leq \frac{(t-1)L^2}{8\beta^2}$. By Pinsker's inequality~\citep{tsybakov2009introduction}, we can bound the total variation (TV) distance as $\operatorname{TV} \left( \pi_{\theta_0}(y_{<t} | x), \, \pi_{cv}(y_{<t} | x) \right) \leq \sqrt{\frac{1}{2} \mathbb{D}_{\mathrm{KL}} \left( \pi_{\theta_0}(y_{<t} | x) \,\|\, \pi_{cv}(y_{<t} | x) \right)} \leq \frac{L\sqrt{t-1}}{4\beta}$.
Meanwhile, by \eqnref{eq:osc-bounds}, we have $|C(x, y_{<t})| \leq \frac{1}{\beta} \cdot \frac{2L \cdot L}{4} = \frac{L^2}{2\beta}$, so $\operatorname{osc}(C(x, y_{<t})) \leq \frac{L^2}{\beta}$. Combining these bounds, changing the prefix distribution from $\pi_{cv}$ to $\pi_{\theta_0}$ incurs an error of at most:
\begin{equation} \label{eq:prefix-shift-error}
    \left| \mathbb{E}_{\mathcal{X}, \pi_{cv}}[C(x, y_{<t})] - \mathbb{E}_{\mathcal{X}, \pi_{\theta_0}}[C(x, y_{<t})] \right| \leq \frac{L^2}{\beta} \cdot \frac{L\sqrt{t-1}}{4\beta} = \frac{L^3 \sqrt{t-1}}{4\beta^2}.
\end{equation}

\noindentparnoline{Expanding the Covariance Term.}
Finally, we can approximate \eqnref{eq:telescoping-joint-diff} by the expected covariance term $C(x, y_{<t})$ over $\mathcal{X}, \pi_{\theta_0}$. Recalling \eqnref{eq:increment-variance}, we can further expand this covariance term as:
\begin{equation} \label{eq:expanded-covariance}
\begin{aligned}
    & \mathbb{E}_{\mathcal{X}, \pi_{\theta_0}}[C(x, y_{<t})] \\
    & = \frac{1}{\beta} \left( \alpha_{c, t} v_{c, t} + \alpha_{v, t} v_{v, t} + \mathbb{E}_{\mathcal{X}, \pi_{\theta_0}} \left[ \Delta U_c(x, y_{\leq t}) \Delta U_v(x, y_{\leq t}) \right] \right).
\end{aligned}
\end{equation}
Combining \eqnref{eq:telescoping-joint-diff}, \eqnref{eq:local-covariance-approx}, \eqnref{eq:prefix-shift-error}, and \eqnref{eq:expanded-covariance}, we apply the resulting approximation to both $\pi_{cv}$ and $\pi_{cv, \mathrm{U}}$ (where $\alpha_{c, t} = \alpha_{v, t} = 1/2$), and subtracting the two expansions \eqnref{eq:expanded-covariance}. Note that the weight-independent cross-covariance term $\mathbb{E}_{\mathcal{X}, \pi_{\theta_0}}[\Delta U_c \Delta U_v]$ cancels out, as well as $\alpha_{v, t} - 1/2 = -(\alpha_{c, t} - 1/2)$, the remaining leading approximation term is precisely $\frac{1}{\beta} \sum_{t=1}^{T} \left( \alpha_{c, t} - \frac{1}{2} \right) (v_{c, t} - v_{v, t})$, while the accumulated error is bounded by $2 \sum_{t=1}^{T} \frac{L^3}{\beta^2} \left( \frac{5}{16} + \frac{\sqrt{t-1}}{4} \right)$, which completes the proof of \eqnref{eq:joint-expansion-bound}.
\end{proof}


\begin{tcolorbox}[kessokustyle=bocchipink]
\RestateTheoremKLMOPDJointObjective*
\end{tcolorbox}

\begin{proof}[Proof of \cref{theorem:kl-mopd-joint-objective}]
By \cref{lemma:kl-variance-relation}, we can write $v_{j, t} = 2\beta^2 d_{j, t} + r_{j, t}$ for each $j$, where $|r_{j, t}| \leq \frac{L_j^3}{6\beta} \leq \frac{L^3}{6\beta}$. Under \mymethod's weights $\alpha_{c, t} = \frac{d_{c, t}}{d_{c, t} + d_{v, t}}$ in \eqnref{eq:weight}, we have $\alpha_{c, t} - \frac{1}{2} = \frac{d_{c, t} - d_{v, t}}{2(d_{c, t} + d_{v, t})}$. Substituting $v_{j, t}$ into the leading term of \eqnref{eq:joint-expansion-bound} at step $t$ yields:
\begin{equation} \label{eq:leading-term-substitution}
\begin{aligned}
    \frac{1}{\beta} \left( \alpha_{c, t} - \frac{1}{2} \right) (v_{c, t} - v_{v, t})
    & = \frac{d_{c, t} - d_{v, t}}{2\beta(d_{c, t} + d_{v, t})} \left( 2\beta^2 (d_{c, t} - d_{v, t}) + r_{c, t} - r_{v, t} \right) \\
    & = \beta \frac{(d_{c, t} - d_{v, t})^2}{d_{c, t} + d_{v, t}} + \frac{d_{c, t} - d_{v, t}}{2\beta(d_{c, t} + d_{v, t})} (r_{c, t} - r_{v, t}).
\end{aligned}
\end{equation}
Since $\left| \frac{d_{c, t} - d_{v, t}}{d_{c, t} + d_{v, t}} \right| \leq 1$ and $|r_{c, t} - r_{v, t}| \leq \frac{L^3}{3\beta}$, the second term in \eqnref{eq:leading-term-substitution} is bounded in absolute value by $\frac{L^3}{6\beta^2}$. Summing \eqnref{eq:leading-term-substitution} over $t = 1, 2, \cdots, T$ and combining the error $\frac{T L^3}{6\beta^2}$ with the bound in \cref{lemma:joint-objective-expansion}, the first leading term in \eqnref{eq:leading-term-substitution} precisely becomes $\beta S_{\mathrm{KL}}$ defined in \eqnref{eq:square-gain}, which is non-negative, and the bound becomes:
\begin{equation}
    \Delta_{\pi_{cv, \mathrm{KL}}, \pi_{cv, \mathrm{U}}} [u_c + u_v]
    \geq \beta S_{\mathrm{KL}} - \frac{L^3}{\beta^2} \left( \frac{19T}{24} + \frac{1}{2} \sum_{t=1}^{T} \sqrt{t-1} \right)
\end{equation}
which precisely matches the desired lower bound in \eqnref{eq:main-bound}. This completes the proof of \cref{theorem:kl-mopd-joint-objective}.
\end{proof}







\end{document}